\documentclass[a4paper,11pt]{article} 
\usepackage{times}  
\usepackage{helvet}  
\usepackage{courier} 
\usepackage[hyphens]{url}  
\usepackage{graphicx} 
\usepackage{bbm}
\usepackage[square,numbers]{natbib}  
\usepackage{caption} 
\DeclareCaptionStyle{ruled}{labelfont=normalfont,labelsep=colon,strut=off} 
\usepackage{fullpage}

\usepackage{tikz}
\usetikzlibrary{shapes}
\usetikzlibrary{patterns}
\usetikzlibrary{decorations.pathreplacing,angles,quotes,positioning,calc}
\usetikzlibrary{matrix}

\usepackage{siunitx}
\usepackage{amsmath,amsthm,amssymb}
\usepackage{graphics}
\usepackage{caption,subcaption}
\usepackage[hidelinks]{hyperref}

\def\s{\mathbf{s}}
\def\x{\mathbf{x}}
\def\y{\mathbf{y}}
\def\z{\mathbf{z}}

\def\v{\mathbf{v}}
\def\V{\mathbf{V}}

\def\f{\mathbf{f}}

\newcommand{\floor}[1]{\left\lfloor #1\right\rfloor}
\newcommand{\ceil}[1]{\left\lceil #1\right\rceil}
\newcommand{\zerotok}[1]{\left\{0..#1\right\}}

\DeclareMathOperator*{\med}{med}

\DeclareMathOperator*{\argmin}{\arg\!\min}

\newcommand{\absolute}[1]{\left\lvert #1 \right\rvert }

\newcommand{\median}[1]{\med \left(#1\right)} 
\usepackage{xcolor}

\newcommand{\mname}[1]{#1}

\newcommand{\profilename}{{three-type}}
\newcommand{\mechanismname}{{Piecewise Uniform}}
\newcommand{\Mechanismname}{{Piecewice Uniform}}

\newtheorem{definition}{Definition}
\newtheorem{theorem}{Theorem}
\newtheorem{lemma}{Lemma}

\theoremstyle{plain}
\newtheorem{claim}{Claim}
\newtheorem{example}{Example}

\title{Truthful Aggregation of Budget Proposals\\ with Proportionality Guarantees\thanks{A preliminary version appeared in the 36th AAAI Conference in Artificial Intelligence.}}
\author{
Ioannis Caragiannis\thanks{ Department of Computer Science,
Aarhus University, \textit{iannis@cs.au.dk}} \and
George Christodoulou\thanks{ Department of Computer Science,
Aristotle University of Thessaloniki, \textit{ gichristo@csd.auth.gr }} \and
Nicos Protopapas\thanks{Department of Computer Science,
University of Liverpool, \textit{N.Protopapas@liverpool.ac.uk}}
}
\date{}

\begin{document}

\maketitle

	\begin{abstract}
          We study a participatory budgeting problem, where a set of
          strategic agents wish to split a divisible budget among
          different projects, by aggregating their proposals on a
          single division.  Unfortunately, the straight-forward rule
          that divides the budget proportionally is susceptible to
          manipulation. In a recent work, \citet{freeman2021journal}
          proposed a class of truthful mechanisms, called \emph{moving
            phantom} mechanisms. Among others, they propose a
          \emph{proportional} mechanism, in the sense that in the
          extreme case where all agents prefer a single project to
          receive the whole amount, the budget is assigned
          proportionally.
		
          While proportionality is a naturally desired property, it is
          defined over a limited type of preference profiles. To
          address this, we expand the notion of proportionality, by
          proposing a quantitative framework which evaluates a budget
          aggregation mechanism according to its worst-case distance
          from the proportional allocation. Crucially, this is defined
          for every preference profile. We study this measure on the class of
          moving phantom mechanisms, and we provide approximation
          guarantees. For two projects, we show that the Uniform
          Phantom mechanism is the optimal among all truthful
          mechanisms. For three projects, we propose a new,
          proportional mechanism which is virtually optimal among all moving
          phantom mechanisms. Finally, we provide impossibility
          results regarding the approximability of moving phantom
          mechanisms.
	\end{abstract}

    \section{Introduction}\label{sec:Intro}
	Participatory budgeting is an emerging democratic process that
        engages community members with public decision-making,
        particularly when public expenditure should be allocated to
        various public projects. Since its initial adoption in the
        Brazilian city of Porto Alegre in the late
        1980s~\cite{cabannes2004participatory}, its usage has been
        spread in various cities across the world. Madrid, Paris, San
        Francisco, and Toronto provide an indicative, but far from
        exhaustive, list of cities that have adopted participatory
        budgeting procedures. See the survey by ~\citet{aziz2021participatory} for
        more examples.
	
        In this paper, we follow the model of~\citeauthor{freeman2021journal}
        in~\cite{freeman2021journal}, where voters are tasked to split an
        exogenously given amount of money among various projects. As
        an illustrative example, consider a city council inquiring the
        residents on how to divide the upcoming year's budget on
        education, among a list of publicly funded schools. Each
        citizen proposes her preferred allocation of the budget and
        the city council uses a suitable aggregation mechanism to
        allocate the budget among the schools.
	
	A natural way to aggregate the proposals is to compute the arithmetic mean for each project and assign to each project exactly that proportion of the budget. Following the above example, consider a community of $80$ thousand residences, wishing to split a budget over $3$ schools. If $20$ thousand residents propose a $(20\%,0\%,80\%)$ division of the budget, other $20$ thousand residents propose a $(40\%,40\%,20\%)$ division and the remaining $40$ thousand residents propose a $(100\%,0\%,0\%)$ division, the budget should be split as $(65\%,10\%,25\%)$ over the $3$ schools. This method, (or variations\footnote{A usual variation is the \emph{trimmed mean} mechanism, where some of the extreme bids are discarded. This is done to discourage a single voter to heavily influence a particular alternative.} of it) is used in practice in economics and sports.  See~\cite{renault2005protecting,renault2011assessing,rosar2015continuous} for some applications. Assigning the budget proportionally comes with some perks: it can be easily described, it is calculated efficiently and it scales naturally to any number of projects.
	
	Unfortunately, allocating the budget proportionally comes also with
    a serious drawback; namely, it is susceptible to
    manipulation. Indeed, consider a simple example with two
    projects and one hundred voters. Fifty voters propose a $(50\%,50\%)$ allocation, while the other fifty voters propose a $(100\%,0\%)$ allocation. Hence, the proportional allocation is
    $(75\%,25\%)$. Assume now, that one voter changes her $(50\%,50\%)$ proposal to $(0\%,100\%)$. This turns the aggregated division to $(74.5\%,25.5\%)$, a division which is closer to the
    $(50\%,50\%)$ proposal that she prefers. Hence, she may have an
    incentive to misreport her most preferred allocation to obtain
    a better outcome, according to her preference.
	
    Truthful mechanisms, i.e. mechanisms nullifying the incentives for strategic manipulation have already been proposed in the literature, for voters with \emph{$\ell_1$} preferences. Under $\ell_1$ preferences (See~\cite{freeman2021journal}), a voter has an ideal division in mind and suffers a disutility equal to the $\ell_1$ distance from her ideal division. \citet{Goel2019} and~\citet{linder08} propose truthful budget aggregation mechanisms that minimize the sum of disutilities for the voters, a quantity known in the literature as the \emph{utilitarian social welfare}. 

    Recently, \citeauthor{freeman2021journal} in~\cite{freeman2021journal}  observed that these mechanisms
    are disproportionately biased towards the opinion of the majority\footnote{As an illustrative example consider an instance where $2k+1$ voters propose budget divisions over two projects. $k+1$ voters assign the whole budget to the first project, while $k$ voters assigns the whole budget to the second project. For this instance, $(100\%,0\%)$ is the unique division that minimizes the sum of disutilities for the voters.}.  This lead
    them to propose the property of \emph{proportionality}. A
    mechanism is \emph{proportional} if, in any input consisted only
    by \emph{single-minded} voters (voters which fully assign the
    budget to a single project), each project receives the proportion
    of the voters supporting that project. They proposed a truthful
    and proportional mechanism, called the \emph{Independent Markets}
    mechanism.
    
    The Independent Markets mechanism belongs to a broader class of
    truthful mechanisms, called \emph{moving phantom mechanisms}. A
    moving phantom mechanism for $n$ voters and $m$ projects, allocates to
    each project the median between the voters' proposals for that
    project and $n+1$ carefully selected \emph{phantom values}. The
    selection of the phantom values is crucial: it ensures both the
    strategy-proofness of the mechanism, as well as its ability to
    return a feasible aggregated division, i.e. that the portions sum
    up to $1$. For example, The Independent Markets mechanism, places
    the $n+1$ phantom values uniformly in the interval $[0,x]$, for
    some $x \in [0,1]$, that guarantees feasibility.
    
    While proportionality is a natural fairness property, it is
    defined only under a limited scope: A proportional mechanism
    guarantees to provide the proportional division {\em only when all
      voters are single-minded}, and provides no guarantee for all
    other inputs. In this paper, we move one step further and we
    address the question: ``How far from the proportional division can
    the outcome of a truthful mechanism be?"
	
    Building on the work of~\citet{freeman2021journal}, we propose a
    more robust measure, and we extend the notion of proportionality
    as follows: Given any input of budget proposals, we define the
    \emph{proportional division} as the coordinate-wise mean of the
    proposals and then we measure the $\ell_1$ distance between the
    outcome of any mechanism and the proportional division. We call
    this metric the \emph{$\ell_1$-loss}. We say that a mechanism is
    $\alpha$-approximate if the maximum $\ell_1$-loss, over all
    preference profiles, is upper bounded by $\alpha\in[0,2]$. {So, in the one extreme $\alpha=0$ implies that the mechanisms always achieves the proportional solution, while in the other extreme $\alpha=2$ has provides no useful implication.}
	
	\subsection{Our Contribution}
	
    	In this paper, we expand the notion of proportionality due to~\citet{freeman2021journal}, by proposing a quantitative
        worst-case measure that compares the outcome of a mechanism
        with the proportional division. We evaluate this measure on
        truthful mechanisms, focusing on the important class of moving
        phantom mechanisms~\cite{freeman2021journal}. Our main
        objective is to design truthful mechanisms with small
        $\alpha$-approximation. We are able to provide effectively optimal
        mechanisms for the case of two and three projects.
    	
    	For the case of two projects, we show that the Uniform Phantom
        mechanism from~\cite{freeman2021journal} is $1/2$-approximate.
        Then, for the case of three projects, we first examine the
        Independent Markets mechanism and we show that this mechanism
        cannot be better than $0.6862$-approximate. We then propose a
        new, proportional moving phantom mechanism which we call the
        \emph{\mechanismname} mechanism which is
        $(2/3+\epsilon)$-approximate, where $\epsilon$ is a small constant\footnote{This constant is at most $10^{-5}$, and arises because we are using a computer-aided proof.}. The analysis of this mechanism is
        substantially more involved than the case of two projects and
        en route to proving the approximation guarantee we
        characterize the instances bearing the maximum $\ell_1$-loss,
        for any moving phantom mechanism.
    	
    	We complement our results by showing matching impossibility
        results: First, we show that there exists no
        $\alpha$-approximate moving phantom mechanism for any
        $\alpha<1-1/m$. This implies that our results for two and
        three projects are essentially best possible, within the family of
        moving phantom mechanisms. Furthermore, we show that no
        $\alpha$-approximate truthful mechanism exists, for
        $\alpha<1/2$, implying that the Uniform Phantom
        mechanism is the best possible among all truthful mechanisms. 
        
        Finally, we turn out attention in cases with large number of projects. We show that our proposed mechanisms and well as the Independent Markets mechanism are at least ($2-\Theta(m^{-1/3})$)-approximate. In addition, we show that any utilitian social welfare maximizing mechanism is ($2-\Theta(m^{-1})$)-approximate.
        
    \subsection{Further Related Work}
	
    Arguably the work closest to our work is the work by \citet{freeman2021journal}. Apart from the Independent Markets
    mechanism they propose and analyze another moving phantom
    mechanism, which maximizes the social welfare and  turned to be equivalent to the truthful mechanisms of~\citet{Goel2019} and~\citet{linder08}, at least up to tie-breaking rules. They also showed that this mechanism is the only mechanism which guarantees Pareto optimal budget divisions.
    
    Both the Uniform Phantom mechanism and the class of moving phantom mechanisms are conceptually related to the \emph{generalized median rules}, from \citeauthor{moulin1980strategy}'s seminal work~\cite{moulin1980strategy}. The generalized median rules are defined on a single-dimensional domain (e.g. the $[0,1]$ line) and are proven to characterize all truthful mechanisms under the mild assumptions of anonymity and continuity. These rules use $n+1$ \emph{phantom values}, where $n$ is the number of voters. For the case of two projects, generalized median rules can be used directly. In fact the Uniform Phantom mechanism, analyzed in Section~\ref{ssec:two_projects}, is a generalized median rule, tailored to our setting. We also note that variants of the Uniform Phantom mechanism appears under different contexts in various papers~\cite{caragiannis16,renault2005protecting,renault2011assessing,aziz2021strategyproof}.

    Unfortunately, generalized median rules cannot be used directly for more than two projects, since there exists no guarantee that the outcome would be a valid division (i.e. the total budget is allocated to the projects). For example, for more than two projects, the Uniform Phantom mechanism returns an outcome which sums to a value at least equal to $1$. 
    The family of moving phantom mechanisms extends the idea of the generalized median rules, by carefully selecting a coordinated set of phantom values, imposing an outcome to be a valid division and retain truthfulness for any number of projects. This family is broad, and it is an open question whether it includes all truthful mechanisms, under the mild assumptions of anonymity, neutrality and continuity. See also the related discussion and references in~\cite{freeman2021journal}.

    Voters with $\ell_1$ preferences are a special case of the well-studied \emph{single-peaked} preferences~\cite{moulin1980strategy} and
    have some precedence in public policy
    literature~\cite{Goel2019}. Recently, a natural utility model, equivalent to $\ell_1$ preferences has been proposed in~\cite{nehring2019resource}.
    We note here that the family of moving phantom mechanism fails to remain strategyproof when the $\ell_1$ preferences assumption is dropped, and replaced with the more general assumption of single-peaked preferences (see the related discussion in~\citet{laraki2021level}).

    The non-truthful mechanism which assigns a budget proportionally has been studied from a strategic point of view in single-dimensional domains (i.e. for two projects).
    \citet{renault2005protecting} and~\citet{renault2007bayesian} explore the equilibria imposed by this mechanism. 
    \citet{rosar2015continuous} compares allocations based on the mean and the median in a single-dimensional domain in a model with incomplete information. In a recent work,~\citet{Puppe21Meanvsmedian} conducted an experimental study between a normalized median rule and the aggregated mean, in a setting similar to ours. Their findings show that under the normalized-median rule, which is not truthful, people where frequently sincere, while under the aggregated median rule the voters' behavior was mainly polarized towards the extremes. 
    	
    In terms of applications, our work falls well under the Participatory Budgeting paradigm. For a survey on Participatory Budgeting, from a mechanism design prospective, the reader is referred to~\cite{aziz2021participatory}. According to the taxonomy therein, our model belongs in the Divisible Participatory Budgeting class. Other examples under the same category, include~\cite{fain2016core,garg2019iterative,aziz2019fair,airiau19,bogomolnaia2005,duddy2015fair,michorzewski2020price}. Among
    others, these works analyze mechanisms with various fairness
    notions, some of which are in the spirit of proportionality. As an example~\citet{michorzewski2020price} quantify the effect of imposing fairness guarantees versus the maximum social possible welfare, under the utility model of dichotomous preferences.
    A large part of the literature concerning Participatory Budgeting covers a model where projects cannot be funded partially, but
    instead are either fully funded or not funded at all. Part of the work of~\citet{Goel2019} is dedicated to this model. Other notable examples include~\cite{benade2020preference,aziz2018proportionally,lu2011budgeted}.
    The problem we tackle here is also related to facility location, a problem with attracted substantial interest by the computational social choice and mechanism design community. For a recent survey, see~\citet{Chan2021FacilityLocationSurvey}.
    
    Finally, our work follows the agenda of approximate mechanism design without money, firstly promoted in~\citet{procaccia2013approximate}. We use an additive approximation measure. Such measures has been used both in the approximation mechanism~\cite{alon2009additive,goemans2006minimum} and the mechanism design~\cite{cai2016facility,serafino2020truthfulness,roughgarden2009quantifying} literature.
    
    \subsection{Roadmap}
    
    The rest of the paper is structured as follows. We begin with preliminary definitions in Section~\ref{sec:preliminaries}. In Section~\ref{sec:upperbounds} we present mechanisms with small approximation for two (Section~\ref{ssec:two_projects}) and three (Section~\ref{ssec:three_projects}). In Section~\ref{sec:lowerbounds} we present impossibility results.

        \section{Preliminaries}\label{sec:preliminaries}
    
	Let $[k]=\{1,...,k\}$ for any $k \in \mathbb{N}$. Let $[n]$ be a set of voters and $[m]$ be a set of projects, for $n \geq 2$ and $m \geq 2$. 
	Let $\mathcal{D}(m)=\{ \x \in [0,1]^m : \sum_{j \in [m]} x_j =1\}$. This set is also known as the \emph{standard simplex}~\cite{boyd2004convex}.
	
	We call a \emph{division} among $m$ projects any tuple $\x \in \mathcal{D}(m)$. Let $d(\x,\y)=\sum_{j \in [m] } \absolute{x_j-y_j}$ denote the $\ell_1$ distance between the divisions $\x$ and $\y$. Voters have structured preferences over budget divisions. Each voter $i \in [n]$ has a most preferred division, her \emph{peak}, $\v^*_i$, and for each division $\x$, she suffers a disutility equal to $d({\v^*_i,\x})$, i.e. the $\ell_1$ distance between her peak $\v^*$ and $\x$.
	
	Each voter $i \in [n]$ reports a division $\v_i$. These divisions form a \emph{preference profile} $\V=(\v_i)_{i \in [n]}$. A \emph{budget aggregation mechanism} $f: \mathcal{D}(m)^n \rightarrow D(m)$  uses the proposed divisions to decide an aggregate division $f(\V)$. A mechanism $f$ is \emph{continuous} when the function  $f: \mathcal{D}(m)^n \rightarrow D(m)$ is continuous. A mechanism is \emph{anonymous} if the output is independent of any voters' permutation and \emph{neutral} if any permutation of the projects (in the voters' proposals) .
	permutes the outcome accordingly.
	
	In this paper, we focus on \emph{truthful} mechanisms, i.e. mechanisms where no voter can alter the aggregated division to her favor, by misreporting her preference.
	
	\begin{definition}[Truthfulness\cite{freeman2021journal}]\label{def:truthfullness}
          A budget aggregation mechanism $f$ is truthful if, for all
          preference profiles $\V$, voters $i$, and divisions $\v^*_i$
          and $\v_i$, $d(f(\V_{-i},\v_i)) \geq d(f(\V_{-i},\v^*_i))$.
	\end{definition}
    
    \noindent A large part of this work is concerned with the class of \emph{moving phantom mechanisms}, proposed in~\cite{freeman2021journal}.
    
	\begin{definition}[Moving phantom mechanisms]\label{def:moving_phantoms}
	Let $\mathcal{Y}=\{y_k: k \in \zerotok{n},\}$ be a family of functions such that, for every $k \in \zerotok{n}$, $y_k:[0,1] \rightarrow [0,1]$ is a continuous, weakly increasing function with $y_k(0)=0$ and $y_k(1)=1$. In addition, $y_0(t) \leq y_1(t) \leq ... \leq y_n(t)$ for every $t \in [0,1]$. The set $\mathcal{Y}$ is called a \emph{phantom system}. For any valid phantom system, a \emph{moving phantom mechanism} $f^\mathcal{Y}$, is defined as follows: For any profile $\V$ and any project $j \in [m]$, 
	
	\begin{equation}\label{eq:def:moving_phantoms}
	f_j^\mathcal{Y}(\V) = \median{ \V_{i \in [n],j},(y_k(t^*))_{k \in \zerotok{n}} }
	\end{equation}
	
	for some   \begin{equation}
	t^* \in \left\{ t: \sum_{j \in [m] } \median{ \V_{i \in [n],j},(y_k(t))_{k \in \zerotok{n}} } =1 \right\} \label{eq:feasibletstar}.
	\end{equation}
	\end{definition}
    
    Each function $y_k(t)$ represents a phantom. At each time $t$, the phantom system returns the values $(y_0(t),...,y_n(t))$. For the particular time $t^*$, these values are sufficient for the sum of the medians to be equal to $1$ (see equation~\ref{eq:feasibletstar}). \citeauthor{freeman2021journal} in~\cite{freeman2021journal} show that one such $t^*$ always exists and, in case of multiple candidate values, the specific choice of $t^*$ does not affect the outcome.
    
    The following theorem states that every budget aggregation mechanism following Definition~\ref{def:moving_phantoms} is truthful.
    
    \begin{theorem}\cite{freeman2021journal}
    Every moving phantom mechanism is truthful.
    \end{theorem}

    Each median in Definition \ref{def:moving_phantoms} can be computed using a sorted array with $2n+1$ slots, numbered from $1$ to $2n+1$. The median value is located in slot $n+1$. Throughout this paper we refer to the slots $1$ to $n$ as the \emph{lower} slots, and $n+2$ to $2n+1$ as the \emph{upper} slots.

	
    \noindent For a given preference profile $\V$, let
    $$\bar{\V}=\left(\frac{1}{n}\sum_{i \in [n]} v_{i,j} \right)_{j \in
      [m]}$$  be the \emph{proportional division}.  A
    \emph{single-minded} voter is a voter $i \in [n]$ such that
    $v_{i,j}=1$ for some project $j \in [m]$.  A budget aggregation
    mechanism is called \emph{proportional}, if for any preference
    profile $\V$ consisted solely by single-minded voters, it holds
    $f(\V)=\bar{\V}$.

    For a given budget aggregation mechanism $f$ and a preference profile $\V$, we define the $f\ell_1$-loss as the $\ell_1$ distance between the outcome $f(\V)$ and the proportional division $\bar{\V}$, i.e.
	
	\begin{align}
	\ell(\V)= d( f(\V),\bar{\V} ) = \sum_{j \in [m] } \absolute{ f_j(\V)-\bar{\V}_j }.
	\end{align}
	
	We say that a budget aggregation mechanism is $\alpha$-approximate when the $\ell_1$-loss for any preference profile is no larger than $\alpha$. We note that no mechanism can be more than $2$-approximate, as the $\ell_1$ distance between any two arbitrary divisions is at most $2$.

    \section{Upper Bounds}\label{sec:upperbounds}
	In this section we present mechanisms with small approximation
    guarantees for $m=2$ and $m=3$. As we will see later on, the guarantee for $2$ projects is optimal for any truthful budget aggregation mechanism. For the case of $3$ projects, our guarantee is practically optimal, when we are confined in the family of moving phantom mechanisms.
    
    \subsection{Two projects}\label{ssec:two_projects}

    For the case of two projects, we focus on the \emph{Uniform
      Phantom} mechanism~\cite{freeman2021journal}, for which we show a $1/2$-approximation. This mechanism is proven to be the unique truthful and proportional mechanism, when aggregating a budget between $2$ projects.
      
     The Uniform Phantom mechanism places $n+1$ phantoms uniformly
    over the $[0,1]$ line, i.e. $$f_j=\median{\V_{i \in
        [n],j},(k/n)_{k \in \zerotok{n}}},$$ for $j \in \{1,2\}$. Later,
    in Theorem \ref{thm:truthful_lowerbound}, we show that $1/2$
    is the best approximation which can be achieved by any truthful mechanism.

    \begin{theorem}\label{thm:Uniform2projects}
	    For $m=2$, the Uniform Phantom mechanism is $1/2$-approximate.
    \end{theorem}

    \begin{proof}
      Let $f$ be the Uniform Phantom mechanism, and let $\V$ be a
      preference profile. Let $f(\V)=(x,1-x)$ and
      $\bar{\V}=(\bar{v},1-\bar{v})$ for some $x \in [0,1]$ and $\bar{v} \in [0,1]$. The loss of the mechanism
      for $\V$ is
    	\begin{align}
    	\ell(\V)=2\absolute{x-\bar{v}}.
    	\end{align}

	Let ${k}\in \zerotok{n}$ be the minimum phantom index such that $ x \leq \frac{{k}}{n} $. This implies that the phantoms with indices $k,...,n$ are located in the slots $n+1$ to $2n+1$. These phantoms are exactly $n+1-{k}$ i.e. exactly ${k}$ voters' reports are located in the same area. Since all values in these slots are at least equal to the median we get that
	\begin{align}\label{eq:m2:1}
	\frac{{k}}{n} \cdot x \leq \bar{v} & \leq \frac{n-{k}}{n} \cdot  x + \frac{{k-1}}{n} + \frac{1}{n}\cdot \mathbbm{1}\left\{x=k/n\right\} \nonumber \\& + \frac{x}{n} \cdot \mathbbm{1}\left\{x<k/n\right\}
	\end{align}
	
	The first inequality holds, since exactly $k$ voters' reports have value at least equal to the median $x$. For the second inequality, we note that exactly $n-k$ voters' reports have value at most $x$, while at least $k-1$ voters' reports can have value at most $1$. If the median is equal to $k/n$, we can safely assume that this is a phantom value, and there should be exactly $k$ values upper bounded by $1$. Otherwise, if the median is strictly smaller than $k/n$, then $x$ should be a voter's report and exactly $k-1$ voters' reports are located in the upper slots. 	
	
	By removing $x$ from both inequalities in \ref{eq:m2:1} we get:
	
	\begin{align}\label{eq:m2:2}
	\frac{{k}}{n} \cdot x - x\leq \bar{v} -x &\leq  \frac{{k-1}}{n} + \frac{1}{n}\cdot \mathbbm{1}\left\{x=k/n\right\} \nonumber \\&+ \frac{x}{n} \cdot \mathbbm{1}\left\{x<k/n\right\} - \frac{k}{n}\cdot x.
	\end{align}
	
	When the median is a phantom value, i.e. $x=\frac{{k}}{n}$, inequalities \ref{eq:m2:2} imply that
	
	\begin{align}
	\absolute{\bar{v}-x} &\leq \max\left\{ x\left(1-\frac{{k}}{n}\right) , \frac{{k}}{n} \left(1-x \right) \right\}   = \frac{k}{n}\left(1-\frac{k}{n}\right) \label{eq:m2:upperbound_x_phantom},
	\end{align} 
	
	which is maximized to for $k=n/2$ to a value no greater than $1/4$. When $x$ is a voter's report, i.e. $\frac{{k}-1}{n} <x<\frac{{k}}{n}$, inequalities \ref{eq:m2:2} imply that
	
	\begin{align}
	\absolute{\bar{v}-x} &\leq \max\left\{ x\left(1-\frac{{k}}{n}\right) , \frac{{k-1}}{n} \left(1-x \right) \right\} \nonumber \\
	& < \max \left\{ \frac{k}{n}\left(1-\frac{k}{n}\right),\frac{{k-1}}{n} \left(1-\frac{k-1}{x} \right) \right\}. \label{eq:m2:upperbound_x_report}
	\end{align} 
	
    Equations~\ref{eq:m2:upperbound_x_phantom} and~\ref{eq:m2:upperbound_x_report} are both upper bounded by $1/4$. The theorem follows. 
    \end{proof}

    	\subsection{Three Projects }\label{ssec:three_projects}
	
	In this subsection we provide a $(2/3+\epsilon)$-approximate truthful
        mechanism for some $\epsilon \leq 10^{-5}$. This mechanism belongs to the family of moving  phantom
        mechanisms, and it is also proportional. In the following, we
        describe the mechanism, and then we prove the approximation
        guarantee. Later, in Theorem
        \ref{thm:moving_phantoms_lower_bound}, we show that $2/3$ is
        the best possible guarantee among the class of moving phantom
        mechanisms.

        \subsubsection{The {\mechanismname} mechanism}

            The {\mechanismname} mechanism uses the phantom system $\mathcal{Y}^{\text{PU}}=\{y_k(t): k \in \zerotok{n}\}$, for which 
            
            \begin{align}\label{eq:NewMechanism:1}
                y_k(t)=\begin{cases}
                0 & \frac{k}{n} < \frac{1}{2} \\
                {\frac {4tk}{n}}-2\,t & \frac{k}{n} \geq \frac{1}{2} \\
                \end{cases}
            \end{align}
            for $t<1/2$, while
            
            \begin{align}\label{eq:NewMechanism:2}
                y_k(t)=\begin{cases}
                    \frac{k(2\,t-1)}{n} &  \frac{k}{n} < \frac{1}{2} \\
                    \frac{k(3-2\,t)}{n} -2 + 2\,t &  \frac{k}{n} \geq \frac{1}{2} \\
                \end{cases}
            \end{align}
            
            for $t \geq 1/2$.  This mechanism belongs to the family of
            moving phantom mechanisms\footnote{ Note that, as presented here, this mechanism does not entirely fit the Definition
              \ref{def:moving_phantoms} since $y_{k}(t)<1$, for all $k
              \in \zerotok{n-1}$. This can fixed easily however, with
              an alternative definition, where all phantom
              functions are shifted slightly to the left and a third
              set of linear functions is added, such that and
              $y_k(1)=1$ for all $k \in \zerotok{n}$. See Section~\ref{appdx:Inclusion} in the Appendix for a detailed explanation.}: each $y_k(t)$
            is a continuous, weakly increasing function, and $y_{k}(t)
            \geq y_{k-1}$(t) for $k \in [n-1]_0$ and any $t \in
            [0,1]$. To distinguish between the two type of phantoms, we call a phantom with index $k < n/2$
            a \emph{black} phantom, and a phantom with index $k \geq n/2$, a \emph{red} phantom (irrespectively of the value o $t$).
            
                    \begin{figure}[htb]
                        \centering
                        \begin{subfigure}[b]{0.4\linewidth}
                            \centering
                            \begin{tikzpicture}[xscale=5]
                            
                            \tikzstyle{vote} = [black, thick]
                            \tikzstyle{phantom} = [red, thick,dashed]
                            
                            \begin{scope}
                            \draw[thick,{Bar[scale=1/2]}-{Bar[scale=1/2]}] (0,-0.05) -- (1,-.05);
                            \node (y0) at (0,-0.35) {0};
                            \node (y1) at (1,-0.35) {1};
                            
                                \fill[fill=black!10] (0,0+0.6) rectangle (1,0.1);
                                \draw[vote] (0.375,0.6)--(0.375,0.1);
                                \draw[vote] (0.125,0.6)--(0.125,0.1);
                                \draw[vote] (0.6250,0.6)--(0.6250,0.1);
                                \draw[vote] (0.4375,0.6)--(0.4375,0.1);
                                \draw[vote] (0.375,0.6)--(0.375,0.1);
                                \node (p1) at (-0.1,0.3) {$j_1$};
                                \draw (0.375-0.02,0.6) rectangle (0.375+0.02,0.1);
                                
                                \fill[fill=black!10] (0,0+1.2) rectangle (1,0.7);
                                \draw[vote] (0.375,1.2)--(0.375,0.7);
                                \draw[vote] (0.5,1.2)--(0.5,0.7);
                                \draw[vote] (0.0625,1.2)--(0.0625,0.7);
                                \draw[vote] (0.5625,1.2)--(0.5625,0.7);
                                \draw[vote] (0.375,1.2)--(0.375,0.7);
                                \draw (0.375-0.02,1.2) rectangle (0.375+0.02,0.7);

                                \node (p2) at (-0.1,0.9) {$j_2$};
                                
                                \fill[fill=black!10] (0,0+1.8) rectangle (1,1.3);
                                \draw[vote] (0.25,1.8)--(0.25,1.3);
                                \draw[vote] (0.375,1.8)--(0.375,1.3);
                                \draw[vote] (0.3125,1.8)--(0.3125,1.3);
                                \draw[vote] (0,1.8)--(0,1.3);
                                \draw[vote] (0.25,1.8)--(0.25,1.3);
                                \draw (0.25-0.02,1.8) rectangle (0.25+0.02,1.3);
                                
                                \node (p3) at (-0.1,1.5) {$j_3$};
                                
                                \draw[phantom,black] (0,1.8)--(0,0.1);
                                \draw[phantom,black] (0,1.8)--(0,0.1);
                                \draw[phantom,black] (0,1.8)--(0,0.1);
                                \draw[phantom] (0.15,1.8)--(0.15,0.1);
                                \draw[phantom] (0.45,1.8)--(0.45,0.1);
                                \draw[phantom] (0.75,1.8)--(0.75,0.1);
                    
                            \end{scope}
                        \end{tikzpicture}
                        \caption{}
                        \label{fig:examples:1}
                        \end{subfigure}\hfill
                        \begin{subfigure}[b]{0.4\linewidth}
                            \centering
                            \begin{tikzpicture}[xscale=5]
                            \tikzstyle{vote} = [black, thick]
                            \tikzstyle{phantom} = [red, thick,dashed]
                            
                            \begin{scope}
                            \draw[thick,{Bar[scale=1/2]}-{Bar[scale=1/2]}] (0,-0.05) -- (1,-0.05);
                            \node (y0) at (0,-0.35) {0};
                            \node (y1) at (1,-0.35) {1};
                            
                            \fill[fill=black!10] (0,0+0.6) rectangle (1,0.1);
                            \draw[vote] (1,0.6)--(1,0.1);
                            \draw[vote] (0.5,0.6)--(0.5,0.1);
                            \draw[vote] (0.3333,0.6)--(0.3333,0.1);
                            \draw[vote] (0,0.6)--(0,0.1);
                            \draw[vote] (0.375,0.6)--(0.375,0.1);
                            \node (p1) at (-0.1,0.3) {$j_1$};
                            \draw (0.375-0.02,0.6) rectangle (0.375+0.02,0.1);
                            
                            \fill[fill=black!10] (0,0+1.2) rectangle (1,0.7);
                            \draw[vote] (0.0,1.2)--(0.0,0.7);
                            \draw[vote] (0.5,1.2)--(0.5,0.7);
                            \draw[vote] (0.5556,1.2)--(0.5556,0.7);
                            \draw[vote] (0.6667,1.2)--(0.6667,0.7);
                            \draw[vote] (0.375,1.2)--(0.375,0.7);
                            \draw (0.4125-0.02,1.2) rectangle (0.4125+0.02,0.7);

                            \node (p2) at (-0.1,0.9) {$j_2$};
                            
                            \fill[fill=black!10] (0,0+1.8) rectangle (1,1.3);
                            \draw[vote] (0.0,1.8)--(0.0,1.3);
                            \draw[vote] (0.0,1.8)--(0.0,1.3);
                            \draw[vote] (0.1111,1.8)--(0.1111,1.3);
                            \draw[vote] (0.25,1.8)--(0.25,1.3);
                            \draw[vote] (0.3333,1.8)--(0.3333,1.3);
                            \draw (0.2125-0.02,1.8) rectangle (0.2125+0.02,1.3);
                            
                            \node (p3) at (-0.1,1.5) {$j_3$};
                            
                            \draw[phantom,black] (0,1.8)--(0,0.1);
                            \draw[phantom,black] (0.1062,1.8)--(0.1062,0.1);
                            \draw[phantom,black] (0.2125,1.8)--(0.2125,0.1);
                            \draw[phantom] (0.4125,1.8)--(0.4125,0.1);
                            \draw[phantom] (0.7063,1.8)--(0.7063,0.1);
                            \draw[phantom] (1,1.8)--(1,0.1);
                            
                            \end{scope}
                            \end{tikzpicture}
                            \caption{}
                            \label{fig:examples:2}
                        \end{subfigure}%
                        \caption{Examples of the {\mechanismname} mechanism, with $5$ voters. The dashed lines correspond to the phantom values, the small rectangles correspond to the medians, and the thick lines correspond to the voters' reports.
                        In the first example $\v_{1}=\v_{2}=(3/8,3/8,1/4)$, $\v_{3}=(1/8,1/2,3/8)$, $\v_{4}=(7/16,9/16,0)$ and $\v_5=(5/8,1/16,5/16)$ and $t^*=3/8$, while
                        the final outcome is $(3/8,3/8,1/4)$.
                        In the second example, $\v_1=(1,0,0)$, $\v_2=(1/2,1/2,0)$, $\v_3=(0,2/3,1/3)$, $\v_4=(1/3,5/9,1/9)$, $\v_5=(3/8,3/8,1/4)$
                        and $t^*=49/64$.}
                        \label{fig:examples}
                    \end{figure}
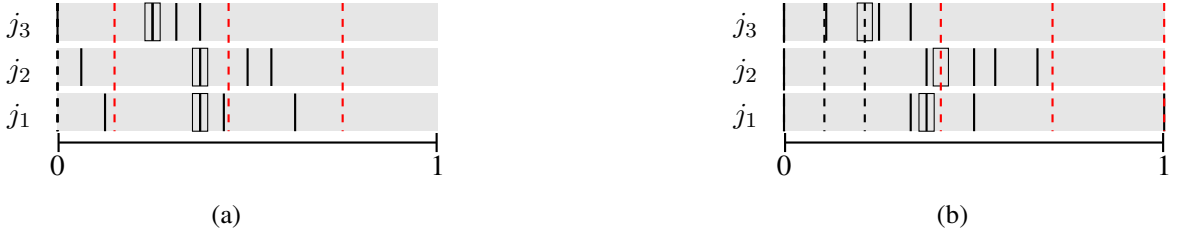

                    This mechanism can been seen as a combination of
                    two different mechanisms: For $t<1/2$, the
                    mechanism uses $n/2$ phantom values equal to $0$,
                    and the rest are uniformly located in
                    $[0,y_n(t)]$. For $t\geq 1/2$, the mechanism
                    assigns half of the phantoms uniformly in
                    $[0,y_{\floor{n/2}}(t)]$, while the rest are
                    uniformly distributed in $[y_{\ceil{n/2}},1]$.
                    See the examples of Figure \ref{fig:examples}, for an illustration.

                    We emphasize here that the {\mechanismname}
                    mechanism admits polynomial time-computation using
                    a binary search algorithm, since
                    $\mathcal{Y}^{\text{PU}}$ is a \emph{piecewise
                      linear phantom system} (see Theorem 4.7 from
                    \cite{freeman2021journal}).

                    We continue by showing that this mechanism is
                    proportional. Note that this does not necessarily
                    need to hold to show the desired approximation
                    guarantee, but it is a nice extra feature of our
                    mechanism.
            
            \begin{theorem}
            	The {\Mechanismname} mechanism is proportional.
            \end{theorem}
            
            \begin{proof}
            	Consider any preference profile which consists exclusively of single-minded voters. Note that by using $t=1$, the phantom with index $k$ has the value $k/n$, for any $k \in \zerotok{n}$. Let that $a_j \in \zerotok{n}$ be the number of $1$-valued proposals  on project $j$. Consequently, $n-a_j$ is the number of $0$-valued proposals. Then the median in each project is exactly the phantom value $a_j/n$, i.e. the proportional allocation. 
            \end{proof}

    	\subsubsection{Upper bound}\label{ssec:upperbound3}

        \paragraph{Overview}
        The analysis for the upper bound is substantially more involved than the analysis for the case of two projects. Here we present an outline of the proof. To help the analysis, we assume at this point that no zero values exists in the aggregated division. We analyze the case where zeros exists in the outcome in Section~\ref{appdx:zeros} in the Appendix.
        
        We first provide a characterization of the worst-case
        preference profiles (i.e. profiles that may yield the maximum
        loss) in Theorem \ref{thm:WorstCaseInstances}. This
        characterization states that essentially all worst-case
        preference profiles belong to a specific family, which we call
        \emph{\profilename} profiles (see Definition
        \ref{def:worst_case_profiles}). The family of {\profilename}
        profiles depends crucially on the moving phantom mechanism used. Given a
        moving phantom mechanism, Lemma \ref{lem:NecessaryConditions}
        characterizes further the family of {\profilename} for that
        mechanism.
        
        We combine Theorem~\ref{thm:WorstCaseInstances} and Lemma~\ref{lem:NecessaryConditions} to build a Non-Linear Program
        (NLP; see Figure~\ref{fig:NLP}) which explores the space
        created by the worst-case instances. Finally we present the
        optimal solution of the NLP in Theorem~\ref{thm:PUmechanism_upper_bound}.  
        
    	\paragraph{Characterization of Worst-Case
          Instances}\label{par:worstcaseinstances}
    	We concentrate on a family of preference profiles which are
        maximal (with respect to the loss) in a local sense: A
        preference profile $\V$ is \emph{locally maximal} if, for all
        voters $i \in [n]$, it holds that $\ell(\V) \geq
        \ell(\V_{-i},\v'_i)$ for any division $\v'_i$. In other words,
        in such profiles, any single change in the voting divisions
        cannot increase the $\ell_1$-loss. Inevitably, any profile
        which may yield the maximum loss belongs to this family, and
        we can focus our analysis on such profiles. Our
        characterization shows that the class of locally maximal
        preference profiles and the class of {\profilename} profiles
        are equivalent with respect to $\ell_1$-loss, for any moving phantom mechanism.
    
    	\begin{definition}[{\profilename} profiles]\label{def:worst_case_profiles}
    		For a given moving phantom mechanism $f$, a preference profile $\V$ is called a \emph{\profilename}\, profile if every voter $i \in [n]$ belongs to one of the following classes:
    		
    		\begin{enumerate}
    			\item \emph{fully-satisfied} voters, where voter $i$ proposes a division equal to the outcome of the mechanism, i.e. $f(\V)=\v_i$, 
    			\item \emph{double-minded} voters,  where voter $i$ agrees with the outcome in one project, i.e. $v_{i,j}=f_j$ for some $j \in [3]$, while $v_{i,j'}=1-f_j$ for some different project $j'$, and
    			\item \emph{single-minded} voters, where $v_{i,j}=1$ for some project $j \in [3]$.
    		\end{enumerate}
    		   
    	\end{definition}
    	
    	\noindent To build intuition, we provide the following example:
    
    	\begin{example}[{\profilename} profile]
    		Consider a moving phantom mechanism $f$, and the preference profile $\V$ with $5$ voters:
    		$\v_1=(1,0,0)$ and $\v_2=(0,0,1)$, which are single-minded voters, $\v_3=(1/2,1/2,0)$, $\v_4=(0,1/4,3/4)$ and $\v_5=(1/2,1/4,1/4)$. Then, if $f(\V)=\v_5$, the preference profile $\V$ is a {\profilename} profile for mechanism $f$. Voter $5$ is a fully-satisfied voter, while voters $3$ and $4$ are double-minded voters.
    	\end{example}
    	
    	In Theorem~\ref{thm:WorstCaseInstances} which follows, we show that for any locally maximal preference profile $\V$, there exists a {\profilename} profile $\hat{\V}$ (not necessarily different than $\V$) for which $\ell(\hat{\V}) \geq  \ell(\V)$. Therefore, we can search for the maximum $\ell_1$-loss by focusing only on the profiles described in Definition~\ref{def:worst_case_profiles}. The following lemma is an important stepping stone for the proof of Theorem~\ref{thm:WorstCaseInstances}. 
    	
    	\begin{lemma}\label{lem:WorstCaseInstances:helper}
    	Let $f$ be a moving phantom mechanism for $m=3$, $\V$ a preference profile and $i \in [n]$, a voter which is neither single-minded, double-minded nor fully-satisfied. Let $\v_i$ be voter's $i$ proposal. Then there exists a division $\v'_i$ such that $\ell(\V_{-i},\v'_i) \geq \ell(\V)$. Furthermore, when $\ell(\V_{-i},\v'_i) = \ell(\V)$  the division $\v'_i$ is double-minded, single-minded or fully-satisfied, and $f(\V)=f(\V_{-i},\v'_i)$.
    	\end{lemma}
    
        	\begin{proof}
		Let $\bar{v}_j = \sum_{i=1}^n v_{i,j}$ for $j \in [3]$ and $\f=f(\V)$.
		We will prove the lemma by constructing the division $\v'_i$. 
		For that, we alter the proposals only on two projects, and we keep the proposal for the third project invariant, in such a way that $\v'_i$ is a valid division. Our first attempt is to strictly increase the loss. When we fail to do that, we create $\v'_i$ in such a way that the loss is not decreasing (comparing $\V$ to $(\V_{-i},\v'_i)$). 
		
		 The following claim, allows us to focus our analysis only on two cases.

		\begin{claim}\label{clm:GuaranteedTwoTypes} 
			Let that $\x,\y,\z$ be valid divisions over $m$ projects, and let that $\x \neq \y$. Then there exists a pair of projects such that either
			i) $x_j \leq \min\{y_j,z_j\}$ and $x_{j'} \geq  \max\{ y_{j'},z_{j'} \}$ or ii) $ y_j \leq x_j \leq z_j $ and $ z_{j'} \leq x_{j'} \leq y_{j'} $.
		\end{claim}
		
		\begin{proof}
			We firstly notice that there exists two projects, say $1$ and $2$, such that $x_1 < y_1$ and $y_2 < x_2$, since $\x \neq \y$. There exist four possible relations between $x_1,x_2$ and $z_1,z_2$: If $x_1 \leq z_1 $ and $ z_2 \leq x_2$ then condition (i) is satisfied. Similarly, if $z_1 \leq x_1$ and $x_2 \leq z_2$, condition (ii) is satisfied. If $x_1 \leq z_1 $ and $ x_2 \leq z_2$, we notice that there exists a different project, say project $3$, such that $z_3 \leq x_3$; otherwise $\sum_{j \in [m]} z_j>1$. We have now two cases, according to the relation between $y_3$ and $x_3$: when $y_3 \leq x_3$, condition (i) holds between projects $1$ and $3$; when $x_3 \leq y_3$, condition (ii)  holds between projects $2$ and $3$. Similar arguments holds for the case $z_1 \leq x_1 $ and $ z_2 \leq x_2$.
		\end{proof}
		
		By Claim \ref{clm:GuaranteedTwoTypes}, we can assume without loss of generality that either i) $f_1 \leq \min\{v_{i,1},\bar{v}_1\}$ and $f_{2} \geq  \max\{ v_{i,2},\bar{v}_{2} \}$ or ii) $ v_{i,1}\leq f_1 \leq v_{i,1} $ and $ \bar{v}_{2} \leq f_{2} \leq \bar{v}_{2}$.
		
		\textbf{Case (i)}: When $v_{i,2}>0$, we can increase the loss as follows: we move $v_{i,1}$ to  $v'_{i,1}=v_{i,1}+\epsilon$ and  $v_{i,2}$ to  $v'_{i,2}=v_{i,2}-\epsilon$,  for some $0<\epsilon \leq \min\{ v_{i,2},1-v_{i,1}\}$. This increases $\bar{v}_1$ to $\bar{v}'_1 = \bar{v}_j + \epsilon/n$ and decreases $\bar{v}_2$ to $\bar{v}'_2=\bar{v}_2 - \epsilon/n$. Note that these moves do not affect the outcome of $f$, since no voters' reports move from the lower to the upper slots or vice versa. Also, $\v'_{i,3}=\v_{i,3}$, hence $\bar{v}'_3=\bar{v}_3$. Thus $\ell(\V_{-i,}\v'_i)> \ell(\V) $.
		See also Figure~\ref{fig:increasing_loss_move}.
		
		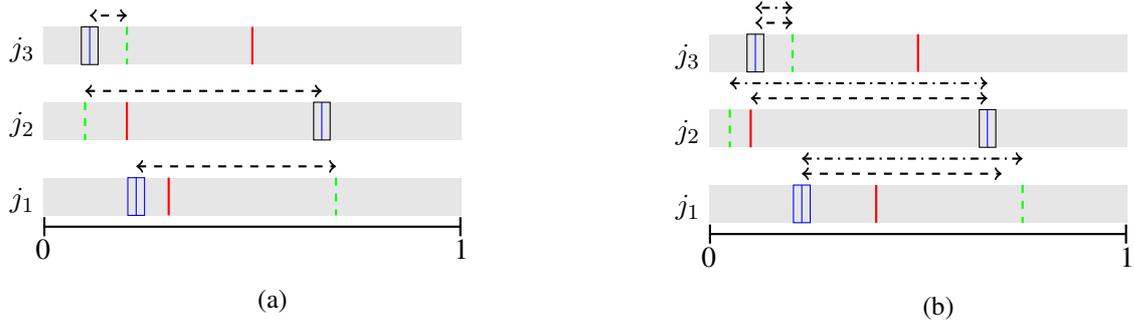
\begin{figure}
    \centering
    \begin{subfigure}{0.45\textwidth}
        \begin{tikzpicture}[xscale=5.5]
            \tikzstyle{voter} = [red, thick]
            \tikzstyle{mean} = [green, thick,dashed]
            \tikzstyle{median} = [blue]
            
            \begin{scope}
            \draw[thick,{Bar[scale=1/2]}-{Bar[scale=1/2]}] (0,-0.05) -- (1,-.05);
            \node (y0) at (0,-0.35) {0};
            \node (y1) at (1,-0.35) {1};
            
            \fill[fill=black!10] (0,0+0.6) rectangle (1,0.1);
            \draw[voter] (0.3,0.6)--(0.3,0.1);
            \draw[mean] (0.7,0.6)--(0.7,0.1);
            \draw[median] (0.2222,0.6)--(0.2222,0.1);
            \draw[median] (0.2222-0.02,0.6) rectangle (0.2222+0.02,0.1);
            \node (p1) at (-0.05,0.3) {$j_1$};
            \draw[thick,<->,dashed] (0.2222,0.75) -- (0.7,0.75);

            \fill[fill=black!10] (0,0+1.6) rectangle (1,1.1);
            \draw[voter] (0.2,1.6)--(0.2,1.1);
            \draw[mean] (0.1,1.6)--(0.1,1.1);
            \draw[median] (0.666,1.6)--(0.666,1.1);
            \draw (0.6666-0.02,1.6) rectangle (0.666+0.02,1.1);
            \node (p2) at (-0.05,1.3) {$j_2$};
            
            \draw[thick,<->,dashed] (0.1,1.75) -- (0.666,1.75);
            
            \fill[fill=black!10] (0,0+2.6) rectangle (1,2.1);
            \draw[voter] (0.5,2.6)--(0.5,2.1);
            \draw[mean] (0.2,2.6)--(0.2,2.1);
            \draw[median] (0.1111,2.6)--(0.1111,2.1);
            \draw (0.1111-0.02,2.6) rectangle (0.1111+0.02,2.1);
            \node (p3) at (-0.05,2.3) {$j_3$};
            
            \draw[thick,<->,dashed] (0.2,2.75) -- (0.1111,2.75);
            
            \end{scope}
            \end{tikzpicture}
            \caption{}
            \label{fig:increasing_loss_move:1}
    \end{subfigure}\hfill
    \begin{subfigure}{0.45\textwidth}
        \begin{tikzpicture}[xscale=5.5]
            \tikzstyle{voter} = [red, thick]
            \tikzstyle{mean} = [green, thick,dashed]
            \tikzstyle{median} = [blue]
            
            \begin{scope}
            \draw[thick,{Bar[scale=1/2]}-{Bar[scale=1/2]}] (0,-0.05) -- (1,-.05);
            \node (y0) at (0,-0.35) {0};
            \node (y1) at (1,-0.35) {1};
            
            \fill[fill=black!10] (0,0+0.6) rectangle (1,0.1);
            \draw[voter] (0.4,0.6)--(0.4,0.1);
            \draw[mean] (0.75,0.6)--(0.75,0.1);
            \draw[median] (0.2222,0.6)--(0.2222,0.1);
            \draw[median] (0.2222-0.02,0.6) rectangle (0.2222+0.02,0.1);
            \node (p1) at (-0.05,0.3) {$j_1$};
            \draw[thick,<->,dashed] (0.2222,0.75) -- (0.7,0.75);
            \draw[thick,<->,dashdotted] (0.2222,0.95) -- (0.75,0.95);

            \fill[fill=black!10] (0,0+1.6) rectangle (1,1.1);
            \draw[voter] (0.1,1.6)--(0.1,1.1);
            \draw[mean] (0.05,1.6)--(0.05,1.1);
            \draw[median] (0.666,1.6)--(0.666,1.1);
            \draw (0.6666-0.02,1.6) rectangle (0.666+0.02,1.1);
            \node (p2) at (-0.05,1.3) {$j_2$};
            
            \draw[thick,<->,dashed] (0.1,1.75) -- (0.666,1.75);
            \draw[thick,<->,dashdotted] (0.05,1.95) -- (0.666,1.95);
            
            \fill[fill=black!10] (0,0+2.6) rectangle (1,2.1);
            \draw[voter] (0.5,2.6)--(0.5,2.1);
            \draw[mean] (0.2,2.6)--(0.2,2.1);
            \draw[median] (0.1111,2.6)--(0.1111,2.1);
            \draw (0.1111-0.02,2.6) rectangle (0.1111+0.02,2.1);
            \node (p3) at (-0.05,2.3) {$j_3$};
            
            \draw[thick,<->,dashed] (0.2,2.75) -- (0.1111,2.75);
            \draw[thick,<->,dashdotted] (0.2,2.95) -- (0.1111,2.95);
            
            \end{scope}
            \end{tikzpicture}
            \caption{}
            \label{fig:increasing_loss_move:2}
    \end{subfigure}
    \caption[Proof of Lemma~\ref{lem:WorstCaseInstances:helper}: loss increasing move]{An example (Figure ~\ref{fig:increasing_loss_move:1}) where the loss can be increased by a single change in voter's $i$ proposed division. Note that voter $i$ is neither fully-satisfied, single-minded nor double-minded. Voter's $i$ proposals are depicted with solid vertical lines, the mean with dashed vertical lines and the outcome of the mechanism is depicted with vertical solid lines inside a rectangle. By moving the proposals (Figure~\ref{fig:increasing_loss_move:2}) of voter $i$ away for the median, the loss strictly increases. }
    \label{fig:increasing_loss_move}
\end{figure}

		If $v_{i,2}=0$, $i$ can be transformed to a single or double-minded voter, without decreasing the loss.
		Note that $v_{i,1}=1-v_{i,3}$. When $v_{i,3}\leq f_3$, $v_{i,3}$ moves to $v'_{i,3}=0$ and $v_{i,1}$ moves to $v'_{i,1}=1$ to create a single-minded division. When $v_{i,3} > f_3$, we can move $v_{i,3}$ to $v'_{i,3}=f_3$ and $v_{i,1}$ to $v'_{i,1}=1-f_3$ (note that $1-v_{i,3} < 1-f_3$) to create a double-minded division $\v'_i$.
		In any case, $v'_{i,1}=v_{i,1}+\epsilon$ and $v'_{i,3}=v_{i,3}-\epsilon$, where $\epsilon= v_{i,3}\cdot\mathbbm{1}\left\{ v_{i,3}\leq f_3 \right\} + (v_{i,3}-f_3)\cdot\mathbbm{1}\left\{ v_{i,3}> f_3 \right\}$.
		Also, $f(\V_{-i,},\v'_i)=f(\V)$ and $\bar{v}'_{2}=\bar{v}_{2}$. Thus,

		\begin{align}
		\ell(\V_{-i,}\v'_i)	&= \bar{v}_1 + \epsilon -f_1 + f_2-\bar{v}_2 \nonumber  + \absolute{\bar{v}_3  -  \epsilon - f_3} \nonumber \\
		& \geq \ell(\V). \nonumber
		\end{align}

		\noindent The inequality holds due to $\absolute{x}\geq x$ for $x \in \mathbb{R}$. See also Figure~\ref{fig:preserving_loss_move}.
		
		\begin{figure}[hb]
    \centering
    \begin{subfigure}{0.45\textwidth}
        \begin{tikzpicture}[xscale=5.5]
            \tikzstyle{voter} = [red, thick]
            \tikzstyle{mean} = [green, thick,dashed]
            \tikzstyle{median} = [blue]
            
            \begin{scope}
            \draw[thick,{Bar[scale=1/2]}-{Bar[scale=1/2]}] (0,-0.05) -- (1,-.05);
            \node (y0) at (0,-0.35) {0};
            \node (y1) at (1,-0.35) {1};
            
            \fill[fill=black!10] (0,0+0.6) rectangle (1,0.1);
            \draw[voter] (0.5,0.6)--(0.5,0.1);
            \draw[mean] (0.7,0.6)--(0.7,0.1);
            \draw[median] (0.2222,0.6)--(0.2222,0.1);
            \draw[median] (0.2222-0.02,0.6) rectangle (0.2222+0.02,0.1);
            \node (p1) at (-0.05,0.3) {$j_1$};
            \draw[thick,<->,dashed] (0.2222,0.75) -- (0.7,0.75);

            \fill[fill=black!10] (0,0+1.6) rectangle (1,1.1);
            \draw[voter] (0.0,1.6)--(0.0,1.1);
            \draw[mean] (0.1,1.6)--(0.1,1.1);
            \draw[median] (0.666,1.6)--(0.666,1.1);
            \draw (0.6666-0.02,1.6) rectangle (0.666+0.02,1.1);
            \node (p2) at (-0.05,1.3) {$j_2$};
            
            \draw[thick,<->,dashed] (0.1,1.75) -- (0.666,1.75);
            
            \fill[fill=black!10] (0,0+2.6) rectangle (1,2.1);
            \draw[voter] (0.5,2.6)--(0.5,2.1);
            \draw[mean] (0.2,2.6)--(0.2,2.1);
            \draw[median] (0.1111,2.6)--(0.1111,2.1);
            \draw (0.1111-0.02,2.6) rectangle (0.1111+0.02,2.1);
            \node (p3) at (-0.05,2.3) {$j_3$};
            
            \draw[thick,<->,dashed] (0.2,2.75) -- (0.1111,2.75);
            
            \end{scope}
            \end{tikzpicture}
            \caption{}
            \label{fig:preserving_loss_move:1}
    \end{subfigure}\hfill
    \begin{subfigure}{0.45\textwidth}
        \begin{tikzpicture}[xscale=5.5]
            \tikzstyle{voter} = [red, thick]
            \tikzstyle{mean} = [green, thick,dashed]
            \tikzstyle{median} = [blue]
            
            \begin{scope}
            \draw[thick,{Bar[scale=1/2]}-{Bar[scale=1/2]}] (0,-0.05) -- (1,-.05);
            \node (y0) at (0,-0.35) {0};
            \node (y1) at (1,-0.35) {1};
            
            \fill[fill=black!10] (0,0+0.6) rectangle (1,0.1);
            \draw[voter] (0.2222,0.6)--(0.2222,0.1);
            \draw[mean] (0.65,0.6)--(0.65,0.1);
            \draw[median] (0.2222,0.6)--(0.2222,0.1);
            \draw[median] (0.2222-0.02,0.6) rectangle (0.2222+0.02,0.1);
            \node (p1) at (-0.05,0.3) {$j_1$};
            \draw[thick,<->,dashed] (0.2222,0.75) -- (0.7,0.75);
            \draw[thick,<->,dashdotted] (0.2222,0.95) -- (0.65,0.95);

            \fill[fill=black!10] (0,0+1.6) rectangle (1,1.1);
            \draw[voter] (0.0,1.6)--(0.0,1.1);
            \draw[mean] (0.1,1.6)--(0.1,1.1);
            \draw[median] (0.666,1.6)--(0.666,1.1);
            \draw (0.6666-0.02,1.6) rectangle (0.666+0.02,1.1);
            \node (p2) at (-0.05,1.3) {$j_2$};
            
            \draw[thick,<->,dashed] (0.1,1.75) -- (0.666,1.75);
            \draw[thick,<->,dashdotted] (0.1,1.95) -- (0.666,1.95);
            
            \fill[fill=black!10] (0,0+2.6) rectangle (1,2.1);
            \draw[voter] (0.8888,2.6)--(0.8888,2.1);
            \draw[mean] (0.3,2.6)--(0.3,2.1);
            \draw[median] (0.1111,2.6)--(0.1111,2.1);
            \draw (0.1111-0.02,2.6) rectangle (0.1111+0.02,2.1);
            \node (p3) at (-0.05,2.3) {$j_3$};
            
            \draw[thick,<->,dashed] (0.2,2.75) -- (0.1111,2.75);
            \draw[thick,<->,dashdotted] (0.3,2.95) -- (0.1111,2.95);
            
            \end{scope}
            \end{tikzpicture}
            \caption{}
            \label{fig:preserving_loss_move:2}
    \end{subfigure}
    \caption[Proof of Lemma~\ref{lem:WorstCaseInstances:helper}: loss preserving move]{An example where the loss cannot be increased by a single change in voter's $i$ proposed division, without changing the outcome of the mechanism. Note that voter $i$ is neither fully-satisfied, single-minded nor double-minded. The voters proposals are depicted with solid vertical lines, the mean with dashed vertical lines and the outcome of the mechanism is depicted with vertical solid lines inside a rectangle. We transform the voter to a double minded voter, while the loss is preserved. }
    \label{fig:preserving_loss_move}
\end{figure}
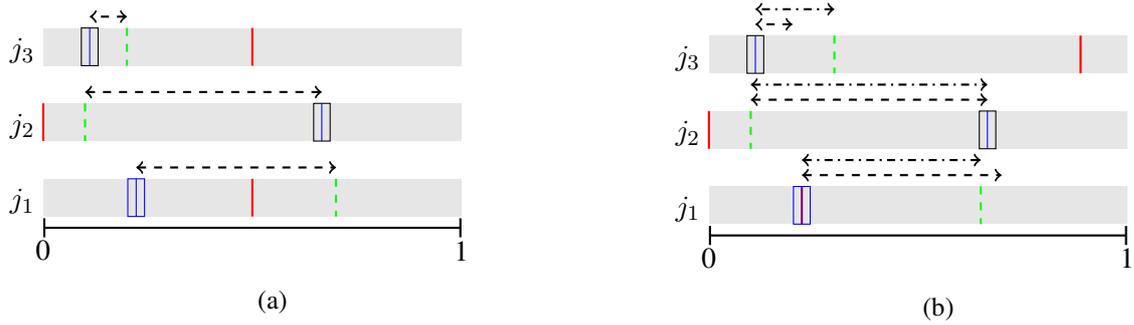
        
		\textbf{Case (ii)}: Recall that $ \bar{v}_1 \leq f_{1} \leq v_{i,1}$ and  and $v_{i,2} \leq f_{2} \leq \bar{v}_{2}$.
		If there exists some $0<\epsilon \leq \min\{v_{i,1}-f_1 ,f_{2} - v_{i,2} \}$ we can strictly increase the loss by $2\epsilon/n$ by moving $v_{i,1}$ to $v'_{i,j}-\epsilon$ and $v'_{i,2}=v_{i,2}+\epsilon$.
		When no such $\epsilon$ exists, either $f_1=v_{i,1}$ or $f_2=v_{i,2}$ (note that this cannot happen for both projects; this would lead to $\v_i=\f$).
		Assume that $v_{i,1}=f_1$. We firstly note that for  project $3$, $v_{i,3}\geq f_{3}$; otherwise, $\sum_{j \in [3]}{v_{i,j}}<1$. We also note that $v_{i,2}< f_{2}$ and $v_{i,3}>f_{3}$, otherwise $\v_i=\f$.
		We will transform $\v_i$ to a fully-satisfied voter, proposing $\v'_i$, without decreasing the loss. 
	    This is done by increasing $v_{i,2}$ to $v'_{i,2}=f_2$ and decreasing $v'_{i,3}$ to $v'_{i,3}=f_3$. Hence $\bar{v}'_2= \bar{v}_2 + \frac{f_2-v_{i,2}}{n}$ and $\bar{v}'_3= \bar{v}_3 - \frac{v_{i,3}-f_3}{n}=\bar{v}_3 - \frac{f_{3}-v_{i,2}}{n}$ (recall that $v_{i,1}=f1$). Note that $f(\V_{-i,},\v'_i)=f(\V)$ and $\bar{v}'_1=\bar{v}_1$. Let $\epsilon=\bar{v}_2 + \frac{f_2-v_{i,2}}{n}$ and 
	  
	    \begin{align}
		\ell(\V_{-i,}\v'_i)&= \bar{v}_1 - f_1 + f_2 + \epsilon -\bar{v}_2 \nonumber  + \absolute{\bar{v}_3  -  \epsilon - f_3} \nonumber \\
		& \geq \ell(\V). \nonumber
		\end{align}

	    \noindent A symmetric argument holds when $v_{i,2}=f_{2}$.
\end{proof}
    
    	\begin{theorem}\label{thm:WorstCaseInstances}
    		Let $f$ be a moving phantom mechanism for $m=3$ and let $\V$ be a locally maximal preference profile, i.e. $\ell(\V) \geq \ell(\V_{-i},\v'_i)$, for any $i \in [n]$ and any division $\v'_i$. Then, there exists a {\profilename} profile $\hat{\V}$ such that $\ell(\hat{\V}) \geq \ell(\V)$.
    	\end{theorem}
    	
    	\begin{proof}
    		 Let $S$ denote the set of single-minded, double-minded or fully-satisfied voters (for mechanism $f$ and for profile $\V$) and let $\bar{S}=[n]\setminus S$. 
    		 
    		 If $\bar{S}=\emptyset$, $\V$ is a {\profilename} profile, hence $\hat{\V}=\V$ and the theorem holds trivially. Otherwise, let $ i \in \bar{S}$. By Lemma \ref{lem:WorstCaseInstances:helper}, we know that we can transform $\v_i$ to $\v'_i$ such that either (a) $\ell(\V_{-i},\v'_i) > \ell(\V)$ or (b) $i$ becomes a double-minded, single-minded or fully-satisfied voter, $f(\V)=f(\V_{-i},\v'_i)$ and $\ell(\V)=\ell(\V_{-i},\v'_i)$. When (a) holds, clearly profile $\V$ is not locally maximal. Hence, we can assume that (b) holds for all voters in $\bar{S}$ and we can create $\hat{\V}$ by transforming all voters in $\bar{S}$ to single-minded, double-minded or fully satisfied, one-by-one. By Lemma \ref{lem:WorstCaseInstances:helper}, both the outcome and the loss stay invariant in each transformation. Hence, $f(\hat{\V})=f(\V)$ and $\ell(\hat{\V})=\ell(\V)$. The theorem follows.
    		\end{proof}

                From now on, we turn our attention on {\profilename} profiles, and 
                in the following we define variables to describe them. A {\profilename} profile $\V$ can be presented 
                using $13$ independent variables:
	\begin{itemize}
		\item $\x=(x_1,x_2,x_3)$, the division of the fully satisfied voters,
		\item $a_1,a_2,a_3$, three integer variables counting the single-minded voter towards each project,
		\item $b_{1,2},b_{1,3},b_{2,1},b_{2,3},b_{3,1},b_{3,2}$, six integer variables counting the double-minded voters (e.g. $b_{2,1}$ counts the voters proposing $(1-x_2,x_2,0)$),
		\item and the total number of voters $n$.
	\end{itemize} 

	We also use $A=\sum_{j \in [3] } a_j$ and $B=\sum_{j,k \in [3],k\neq j } b_{k,j}$ to count the single-minded and the double-minded voters, respectively. Consequently, the number of fully satisfied voters is  $C=n-A-B$. These profiles can have at most $8$ distinct voters' reports: values $x_1$, $x_2$ and $x_3$, from fully-satisfied and double-minded voters, values $1-x_1$, $1-x_2$ and $1-x_3$ which we call \emph{complementary values} from the double-minded voters and, reports with values equal to $1$ and $0$. Note that, apart from values $0$ and $1$, in project $1$ we can find values $x_1$, $1-x_2$ and $1-x_3$, in project $2$ values $x_2$, $1-x_1$ and $1-x_3$ and finally in project $3$ the values $x_3$, $1-x_1$ and $1-x_2$.
	
	
    Recall that Definition~\ref{def:worst_case_profiles} demands that $f(\V)=\x$. To ensure this, we prove the following lemma. Note that we assume that $x_j>0$, for all $j \in [3]$ at this point.

 	\begin{lemma}\label{lem:NecessaryConditions}
 	    Let that $x_j>0$, for all $j \in [3]$. Let $z_j=a_j + \sum_{k \in [3]\setminus\{j\}}{b_{k,j}}$ and $q_j = \sum_{k \in [3]\setminus \{j\}} b_{j,k}$.
		For any moving phantom mechanism $f$, defined by the phantom system $\mathcal{Y}=\{y_k(t): k \in \zerotok{n} \}$, and any {\profilename} profile $\V$, then $f(\V)=\x$ if and only if 
		
		\begin{align}\label{eq:lem:NecessaryConditions}
		     y_{z_j}(t^*) \leq  x_j \leq y_{z_j+q_j+C}(t^*)
		\end{align}

		for any \[t^* \in \left\{ t: \sum_{j \in [m] } \median{ \V_{i \in [n],j},(y_k(t))_{k \in \zerotok{n}} } =1 \right\}.\]
	\end{lemma}
	
	\begin{proof}
	     First note that for $x_j>0$ for all $j \in [3]$, all complementary values $1-x_1$, $1-x_2$ and $1-x_3$ are located in the upper slots. Assume otherwise, that there exists some complementary value, say $1-x_2$ such that $1-x_2 \leq x_1$. Then $1 \leq x_1 + x_2$, which is not possible when $x_3>0$. In addition, all $1$-valued voters' reports should be located in the upper slots, while all $0$-valued voter reports should be located in the lower slots. Note also that $z_j=a_j + \sum_{k \in [3]\setminus\{j\}}{b_{k,j}}$ counts exactly the voters' reports located in the upper slots.
	    
		(if direction) 
		Let $\V$ be a {\profilename} profile and let $f(\V)=\x$ for some $t^* \in [0,1]$.
		Assume, for the sake of contradiction that $y_{z_j}(t^*)>x_j$. This implies that the $n+1-z_j$ phantoms with indices $z_j,...,n$ are located in the upper slots (i.e. the $n$ higher slots). Since at least $z_j$ voters' reports are also be located in the upper slots there exists at least $n+1$ values for $n$ slots. A contradiction.
		Suppose now that that $y_{z_j+q_j+C}(t^*)<f_j(\V)$. This implies that $z_j+C+1$ phantom values (those with indices $0,...,z_j+C)$) are located in the lower slots (i.e. the $n$ lower slots). The voters' reports with value $0$ must be also be located in the lower slots, since $0<x_j$ for any $ \in [3]$.
        There are exactly $ A + B - q_j -z_j   = n-C- q_j -z_j$ such values.  Hence at least $n+1$ values should be located in the lower slots. A contradiction.
         		
        (only if direction ) Let that inequalities \ref{eq:lem:NecessaryConditions} hold and let $\V$ be  a {\profilename} profile. Assume for the sake of contradiction that there exists a project $j \in [3]$ such that $ f_j(\V) < x_j$. Hence, the $C$ values of the possibly fully satisfied voters, plus the $q_j$ values equal to $x_j$ should be located in the upper slots. The $1$-valued voters' reports, which count to $a_1$ and the complementary values, which count to $\sum_{k \in [3]\setminus\{j\}}{b_{k,j}}$ are also located in the upper slots. These count to $z_j+q_j+C$ values. Furthermore, since $x_j \leq y_{z_j+q_j+C}(t^*)$ , another $n-C-z_j-q_j+1$ phantom values should be located in the upper slots. A contradiction. 
        
        Similarly, assume for the sake of contradiction that there exists a profile $j \in [3]$ such that $f_j(\V)>x_j$. Then the $C+q_j$ values should be located in the lower slots, along with the $n-C-q_j-z_j$ voters' reports equal to $0$. Since $y_{z_j}(t^*) \leq x_j$ ,  $z_j+1$ phantom values are also located in the lower slots. These are in total $n+1 $ values, a contradiction. 
    \end{proof}

    \begin{figure}[ht]
    \begin{subfigure}{1\textwidth}
        \begin{tikzpicture}[xscale=5]
            \draw (0,0) -- (3,0);
            \draw (0,0.5) -- (3,0.5);
            \draw[decoration={brace,mirror,raise=10pt},decorate] (0,0) --node[below=10pt] {$n$} (1.5-0.0625,0);
            \draw[decoration={brace,mirror,raise=10pt},decorate] (1.5+0.0625,0) --node[below=10pt] {$n$} (3,0);
            \draw[decoration={brace,raise=10pt},decorate] (2.7,0.5) --node[above=10pt] {$a_1$} (3,0.5);
            \draw[decoration={brace,raise=10pt},decorate] (2.2,0.5) --node[above=10pt] {$b_{2,1}$} (2.4,0.5);
            \draw[decoration={brace,raise=10pt},decorate] (1.8,0.5) --node[above=10pt] {$b_{3,1}$} (2,0.5);
            \draw[decoration={brace,raise=10pt},decorate] (1.5-0.0625,0.5) --node[above=10pt] {$C+q_1$} (1.78,0.5);
            \draw[decoration={brace,raise=10pt},decorate] (0,0.5) --node[above=10pt] {$ n-C-q_1 -z_1$} (0.425,0.5);
            \filldraw[pattern=vertical lines] (0,0) rectangle (0.425,0.5);
            \fill[color=gray] (0.425,0) rectangle (1.5-0.0625,0.5);
            \filldraw[thick,double,pattern=horizontal lines, pattern color= green] (1.5-0.0625,0) rectangle (1.5+0.0625,0.5);
            \fill[pattern=horizontal lines, pattern color= green] (1.5+0.0625,0) rectangle (1.8,0.5);
            \fill[pattern=north east lines, pattern color=red] (1.8,0) rectangle (2,0.5);
            \fill[color=gray] (2,0) rectangle (2.2,0.5);
            \fill[pattern=north west lines, pattern color=blue] (2.2,0) rectangle (2.4,0.5);
            \fill[color=gray] (2.4,0) rectangle (2.7,0.5);
            \filldraw[pattern=vertical lines,pattern color=brown] (2.7,0) rectangle (3,0.5);
        \end{tikzpicture}%
        \caption{}
        \label{fig:ThreeTypeValid:a}
    \end{subfigure}
    \begin{subfigure}{0.99\textwidth}
        \begin{tikzpicture}[xscale=5]
            \draw (0,0) -- (3,0);
            \draw (0,0.5) -- (3,0.5);
            \draw[decoration={brace,mirror,raise=10pt},decorate] (0,0) --node[below=10pt] {$n$} (1.5-0.0625,0);
            \draw[decoration={brace,mirror,raise=10pt},decorate] (1.5+0.0625,0) --node[below=10pt] {$n$} (3,0);
            \draw[decoration={brace,raise=10pt},decorate] (2.7,0.5) --node[above=10pt] {$a_1$} (3,0.5);
            \draw[decoration={brace,raise=10pt},decorate] (2.2,0.5) --node[above=10pt] {$b_{2,1}$} (2.4,0.5);
            \draw[decoration={brace,raise=10pt},decorate] (1.8,0.5) --node[above=10pt] {$b_{3,1}$} (2,0.5);
            \draw[decoration={brace,raise=10pt},decorate] (1.5-0.3625,0.5) --node[above=10pt] {$C+q_1$} (1.5+0.0625,0.5);
            \draw[decoration={brace,raise=10pt},decorate] (0,0.5) --node[above=10pt] {$ n-C-q_1 -z_1$} (0.425,0.5);
            \fill[gray] (1.5+0.0625,0) rectangle (1.8,0.5);
            \filldraw[pattern=vertical lines] (0,0) rectangle (0.425,0.5);
            \filldraw[thick,double, pattern=horizontal lines, pattern color= green] (1.5-0.0625,0) rectangle (1.5+0.0625,0.5);
            \fill[color=gray] (0.425,0) rectangle (1.5-0.3625,0.5);
            \fill[pattern=horizontal lines, pattern color= green] (1.5-0.3625,0) rectangle (1.5-0.0625,0.5);
            \fill[pattern=north east lines, pattern color=red] (1.8,0) rectangle (2,0.5);
            \fill[color=gray] (2,0) rectangle (2.2,0.5);
            \fill[pattern=north west lines, pattern color=blue] (2.2,0) rectangle (2.4,0.5);
            \fill[color=gray] (2.4,0) rectangle (2.7,0.5);
            \filldraw[pattern=vertical lines,pattern color=brown] (2.7,0) rectangle (3,0.5);
        \end{tikzpicture}%
        \caption{}
        \label{fig:ThreeTypeValid:b}
    \end{subfigure}    
    \caption[Proof of Lemma \ref{lem:NecessaryConditions}: phantom values positioning]{Example for the positioning of phantom and voters reports in project $1$, for a given three-type profile. The $5$ patterned intervals represent the voters' reports. The solid, dark intervals represent the phantom values. The double-lined rectangle in the middle represents the median. Figure~\ref{fig:ThreeTypeValid:a} illustrates the case where the $C+b_{1,2}+b_{1,3}=C+q_1$ voters' reports with values $x_1$ are located in the upper $n+1$ slots. The $z_1=a_1 + b_{2,1} +b_{3,1}$ voters' reports with values $1$ and $1-x_2$ and $1-x_3$ must also be located in the top $n+1$ slots. Figure~\ref{fig:ThreeTypeValid:b} illustrates the other extreme, where $C+q_1$ values equal to the median are located in the lower $n+1$ slots.}
    \label{fig:ThreeTypeValid}
\end{figure}
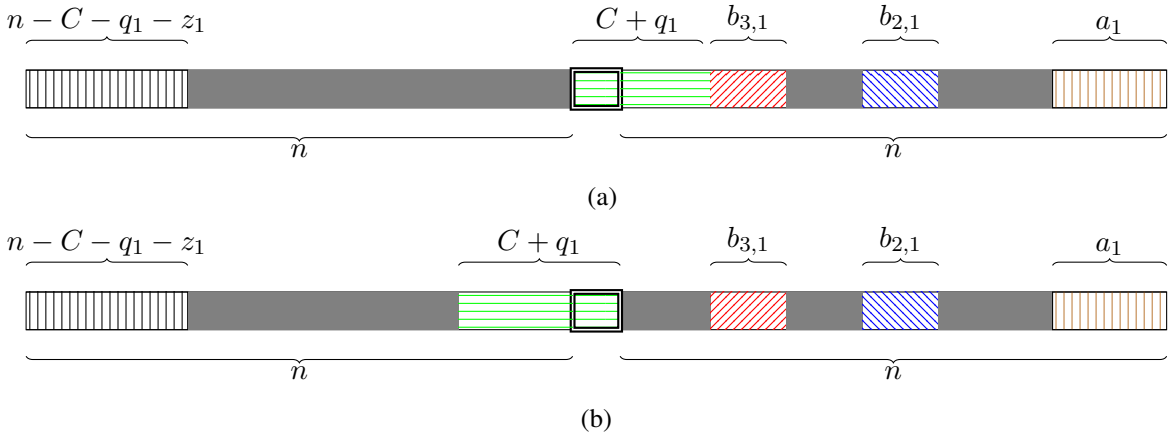
    
            We note that the only-if direction is not required for the
            proof of Theorem \ref{thm:PUmechanism_upper_bound}, but it
            is a nice feature that we include for the sake of
            completeness.

        \paragraph{A Non-Linear Program}
        
        We show that the {\mechanismname} mechanism is $(2/3+\epsilon)$-approximate using a Non-Linear Program. The feasible region of this program is defined by the class of {\profilename} profiles, and we search for the maximum $\ell_1$-loss among them. 
        For simplicity, we firstly normalize some of our variables with $n$. We introduce new variables $\hat{a}_j=a_j/n$ for $j \in [3]$ and $\hat{b}_{j,j'}=b_{j,j'}/n$ for $j,j' \in [3]$, $j \neq j$, and $\hat{C}=C/n$. We also use a relaxed version of the {\mechanismname} phantom mechanism: For every $x \in [0,1]$:
         
         \begin{equation}\label{eq:NewMechanismRelaxed}
        \hat{y}(x,t)=
        \begin{cases}
        0 & 0 \leq  t < \frac{1}{2} \textbf{ and } x \leq \frac{1}{2} \\
        {4tx} - 2t & 0< t < \frac{1}{2} \textbf{ and } x > \frac{1}{2} \\
        {x(3-2t)} -2 + 2t & \frac{1}{2} \leq  t \leq 1 \textbf{ and } x > \frac{1}{2} \\
        {x(2t-1)} & \frac{1}{2} \leq t \leq 1 \textbf{ and } x \leq \frac{1}{2}. \\
        \end{cases} \nonumber
        \end{equation}
         
         
        \noindent To help the presentation, we introduce also variables for the mean for each project $j \in [3]$:
        	\begin{align}
        	\bar{v}_j&=\hat{a}_j + \sum_{k\in[3]\setminus \{j\}} {(1-x_k)\hat{b}_{k,j}}  \nonumber  + x_j{\left(\hat{C}+\sum_{k \in [3] \setminus \{j\}} \hat{b}_{j,k}\right)}.
        	\end{align}

        	\begin{figure}[t]
        
        	\begin{align}
        	\text{maximize}& \sum_{j=1}^3\absolute{ \bar{v}_j -x_j} \label{MIQCP:objective} \\
        	\text{subject to}& && \nonumber \\ 
        	\sum_{j=1}^3 x_j& =1, && \nonumber  \\
        	\hat{A} &=\sum_{j=1}^3 \hat{a}_j, && \nonumber \\
        	\hat{B} &= \sum_{j,k \in [3],j\neq k} \hat{b}_{k,j}, && \nonumber \\
        	\hat{z}_j&= \hat{a}_j + \sum_{k \in [3]\setminus\{j\}}{\hat{b}_{k,j}}, && \forall j \in [3] \nonumber \\
        	\hat{q}_j&= \hat{a}_j + \sum_{k \in [3]\setminus\{j\}}{\hat{b}_{j,k}}, && \forall j \in [3] \nonumber \\
        	 x_j & \geq \hat{y}\left(\hat{z}_j, t^* \right), && \forall j \in [3]\label{ineqNew:lowerbound} \\
        	 x_j &\leq \hat{y}\left(\hat{C} + \hat{q}_j+\hat{z}_j,t^* \right),&&  \forall j \in [3] \label{ineqNew:upperbound} \\
             \hat{A}+\hat{B} & \leq 1, && \nonumber \\
        	x_j &\geq 0, \hat{a}_j \geq 0 ,&& \forall j \in [3] \nonumber \\
        	 \hat{b}_{k,j} &\geq 0 ,&& \forall j,k \in [3],j\neq k \nonumber \\
        	 0 &\leq t^* \leq 1. \nonumber
        	\end{align}
        	\caption{The Non-Linear Program used to upper bound the maximum $\ell_1$-loss for the {\mechanismname} mechanism.}
        	 \label{fig:NLP}
        	\end{figure}
        
                The Non-Linear Program is presented in Figure
                \ref{fig:NLP}. Inequalities \ref{ineqNew:lowerbound}
                and \ref{ineqNew:upperbound} ensure that we are
                searching over all {\profilename} profiles for the {\mechanismname} mechanism. Crucially,
                any profile which does not meet these two conditions
                cannot have $\x$ as the outcome (see Lemma
                \ref{lem:NecessaryConditions}). Finally, we let the
                program optimize over any  $t^* \in
                [0,1]$. Lemma \ref{lem:NecessaryConditions} ensures
                that any value $t^*$ that satisfies inequalities
                \ref{ineqNew:lowerbound} and \ref{ineqNew:upperbound}
                will return a valid outcome.
        	
        \paragraph{Maximum Loss Computation}
        
        As presented in Figure~\ref{fig:NLP}, the Non-Linear Program is quite complex.
        To compute an upper bound for its maximum value, we break this program into simpler, Quadratic programs, based on $3$ conditions; first, depending on whether $t^*\leq 1/2$ or
        not, second, according to the signs of the $\bar{v}_j -x_j$
        terms on the objective function (in order to remove the
        absolute values), and finally, according to the types of the
        phantoms enclosing the medians.
        
        To deal with the signs of the $\bar{v}_j -x_j$ terms, we define \emph{sign patterns}, as tuples in $\{+,-\}^3$. For example the sign pattern $(+,+,-)$ shows that $\bar{v}_1 \geq x_1$ and $\bar{v}_2 \geq x_2$, while $x_3 \geq \bar{v}_3$. We note that we cannot have the same sign is all projects, unless the loss is equal to $0$.
        Assume otherwise, that there exists a sign pattern with the same sign in all projects, say $(+,+,+)$. Then, $\sum_{j=1}^3 (\bar{v}_j - x_j) = \sum_{j=1}^3 \bar{v}_j- \sum_{j=1}^3 x_j=0.$ Hence, we only need to check the patterns $(+,-,-)$ and $(+,+,-)$. 
        
        We need also to address the discontinuities in function $\hat{y}$, with respect to the first argument. For this, we use the tuple $(b,r)$ to distinguish weather the median lies between two red, two black, or between a red an a black phantom. By noting that $\hat{z}_j>1/2$ implies that $\hat{C}+\hat{q}_j+\hat{z}_j>1/2$, we can safely assume that no median is upper bounded by a black phantom and lower bounded by a red phantom, and we define \emph{phantom patterns}, as tuples in  $\{(b,b),(b,r),(r,r)\}^{3}$. We build a quadratic program for each phantom pattern.
        
        In total, we end up with $2 \times 2 \times 27=108$ Quadratic Programs with Quadratic Constraints (QPQC). In Figure~\ref{fig:NLP:worst-case} we present in detail on of these programs. Specifically, we present the case where $t>1/2$ with sign pattern $(+,-,-)$ and phantom pattern $((r,r),(b,b),(b,b))$. This quadratic program corresponds to a case yielding highest upper bound for the $\ell_1$-loss.
        
        To prove Theorem~\ref{thm:PUmechanism_upper_bound} which follows, we first compute the maxima for $27$ programs corresponding to the case $t^*>1/2$ and the sign pattern $(+,-,-)$. For the remaining QPQCs, i.e. for $t\geq 1/2$ and the sign pattern $(+,+,-)$ and for $t^*<1/2$ for both sign patterns, we check whether any of these cases can yield loss greater than $2/3$. We present one example in Figure \ref{fig:NLP:++-}. All these programs are infeasible, i.e. no one of these cases yield $\ell_1$-loss greater than $2/3$ plus a computational error term which is upper bounded by $10^{-5}$.

        
        We solve these programs using the Gurobi optimization software~\cite{gurobi}. The solver models our programs as Mixed Integer Quadratic Programs\footnote{note that our QPQCs are non-convex.} and uses the \emph{spatial Branch and Bound Method} (see \cite{liberti2008introduction}) with various heuristics, to return a global maximum, when the program is feasible. The solver computes arithmetic solutions with $10^{-5}$ error tolerance.
        
        \begin{theorem}\label{thm:PUmechanism_upper_bound}
        The {\mechanismname} mechanism is $(2/3+\epsilon)$-approximate, for some $\epsilon \in [0,10^{-5}]$.
        \end{theorem}
        
        \begin{proof}
            Theorem~\ref{thm:WorstCaseInstances} states that the maximum loss for any moving phantom mechanism happens in a {\profilename} profile. The Non-Linear Program in Figure~\ref{fig:NLP} searches for the profile with maximum loss, over all {\profilename} profiles. The latter is guaranteed by Lemma \ref{lem:NecessaryConditions}. We first solve $27$ QPQCs, corresponding to the sign pattern $(+,-,-)$ and $t^*>1/2$. The maximum value is no higher than $2/3+\epsilon$, where $\epsilon$ is due to the error tolerance of the solver. Table~\ref{tab:results} shows analytically the upper bounds computed for each one of the $27$ QPQCs (excluding some symmetric cases).
            
            For the other $81$ QPQCs we check whether any of them yields loss at least $2/3$. For that, we add the constraint $\sum_{j=1}^3 s(j)( \bar{v}_j -x_j) \geq 2/3$, where $s(j)$ denotes the sign for project $j$ according to the sign pattern and we search for any feasible solution. Eventually, no feasible solution exists, i.e. there exists no other preference profile with loss at least $2/3$ plus the computational error imposed by the solver.   
            
            To complete our analysis we need to address the case where the outcome includes at least one $0$ value. We analyze this case in Section~\ref{appdx:zeros} in the Appendix and we show that the $\ell_1$-loss cannot be higher than $1/2+\epsilon$ in this case.
        \end{proof}

    
    \begin{table}[t]
    \centering
    \begin{tabular}{|c|c|c|c|c|}
        \hline
          \textbf{Phantoms} & \textbf{Status} & \textbf{ Loss (lower bound)} &  \textbf{ Loss (Upper bound)} & \textbf{Gap} \\ \hline
            $(b,b),(b,b),(b,b)$ & OPTIMAL & $0.333332$ & $0.333333$ & $\num{1.83e-06}$ \\ 
            $(b,b),(b,b),(b,r)$ & OPTIMAL & $0.357003$ & $0.357007$ & $\num{3.69e-06}$ \\ 
            $(b,b),(b,b),(r,r)$ & OPTIMAL & $0.357003$ & $0.357010$ & $\num{7.38e-06}$ \\ 
            $(b,b),(b,r),(b,r)$ & OPTIMAL & $0.499998$ & $0.500014$ & $\num{1.62e-05}$ \\ 
            $(b,b),(b,r),(r,r)$ & OPTIMAL & $0.499999$ & $0.500010$ & $\num{1.11e-05}$ \\ 
            $(b,b),(r,r),(b,b)$ & OPTIMAL & $0.357004$ & $0.357008$ & $\num{4.27e-06}$ \\ 
            $(b,b),(r,r),(r,r)$ & OPTIMAL & $0.000000$ & $0.000000$ & $\num{0.00e+00}$ \\ 
            $(b,r),(b,b),(b,b)$ & OPTIMAL & $0.666667$ & $0.666667$ & $\num{-2.22e-16}$ \\ 
            $(b,r),(b,b),(b,r)$ & OPTIMAL & $0.666666$ & $0.666668$ & $\num{1.92e-06}$ \\ 
            $(b,r),(b,b),(r,r)$ & OPTIMAL & $0.529134$ & $0.529141$ & $\num{7.08e-06}$ \\ 
            $(b,r),(b,r),(b,r)$ & OPTIMAL & $0.666667$ & $0.666672$ & $\num{5.37e-06}$ \\ 
            $(b,r),(b,r),(r,r)$ & OPTIMAL & $0.500000$ & $0.500006$ & $\num{5.67e-06}$ \\ 
            $(b,r),(r,r),(b,b)$ & OPTIMAL & $0.529134$ & $0.529139$ & $\num{5.40e-06}$ \\ 
            $(b,r),(r,r),(r,r)$ & OPTIMAL & $0.000000$ & $0.000000$ & $\num{0.00e+00}$ \\ 
            $(r,r),(b,b),(b,b)$ & OPTIMAL & $0.666667$ & $0.666669$ & $\num{1.86e-06}$ \\ 
            $(r,r),(b,b),(b,r)$ & OPTIMAL & $0.666666$ & $0.666667$ & $\num{5.70e-07}$ \\ 
            $(r,r),(b,b),(r,r)$ & OPTIMAL & $0.527863$ & $0.527866$ & $\num{3.25e-06}$ \\ 
            $(r,r),(b,r),(b,r)$ & OPTIMAL & $0.666666$ & $0.666667$ & $\num{1.12e-06}$ \\ 
            $(r,r),(b,r),(r,r)$ & OPTIMAL & $0.500000$ & $0.500004$ & $\num{3.87e-06}$ \\ 
            $(r,r),(r,r),(b,b)$ & OPTIMAL & $0.527864$ & $0.527867$ & $\num{2.55e-06}$ \\ 
            \hline
            \end{tabular}
        \caption{The bounds computed by the QPQCs, for $t>1/2$ and the sign pattern $(+,-,-)$. The programs without feasible solutions are not presented, as well as symmetric cases. The lower bound corresponds to the largest loss for a feasible solution computed by the solver. The upper bound corresponds to the smaller non-feasible lower bound computed be the solver. The last column shows the gap between them. Gaps smaller than $10^{-5}$ are insignificant due to the tolerance of the solver.}
        \label{tab:results}
        \end{table}

    \section{Lower Bounds}\label{sec:lowerbounds}
    
    In this section, we provide impossibility results for our proposed measure. Theorem~\ref{thm:truthful_lowerbound} shows that no truthful mechanism can be less than $1/2$-approximate. Theorem~\ref{thm:moving_phantoms_lower_bound} focuses on the class of moving phantom mechanisms and shows that no such mechanism can be less than $(1-1/m)$-approximate. Theorem~\ref{thm:IMlowerbound3} shows that the Independent Markets mechanism is $0.6862$-approximate. Finally, we present lower bounds for large $m$: A combined lower bound of $2-\frac{8}{m^{1/3}}$ for both the Independent Markets mechanism and the Piecewise Uniform mechanism, and a lower bound of $2-\frac{4}{m+1}$ for any mechanism which maximizes the social welfare.
    
    	\subsection{A Lower Bound  for any Truthful Mechanism}

        In the following, we show that truthfulness inevitably admits $\ell_1$-loss  at least $1/2$ in the worst case. We recall that the Uniform Phantom mechanism achieves this bound for $m=2$.
        
        \begin{theorem}\label{thm:truthful_lowerbound}
        	No truthful mechanism can achieve $\ell_1$-loss less than $1/2$.
        \end{theorem}
        
        \begin{proof}
        	Let $f$ be a truthful mechanism over $m$ projects. Consider a profile with $2$ voters $\V=(\v_1,\v_2)$, such that $\v_1=(1,0,...,0)$ and $\v_2=(0,1,0,...,0)$ and let that $f(\V)=(x_1,...,x_m)=\x$ for some $x \in \mathcal{D}(m)$. Consider also the profile $\V'=(\x,\v_2)$. Due to  truthfulness, then $f(\V')=\x$. Assume otherwise, that $f(\V)=\x'\neq \x$; when voter's $1$ peak is at $\x$, i.e. $\v_1^*=\x$, then the disutility for voter $1$ when proposing $\v_1$ is:

        	\begin{align}
        	    d(f(\V),\v_1^*)=d(\x,\x)=0
        	\end{align}
        	
        	while the disutility for voter $1$ when proposing $\v_1^*$ is
        	
        	\begin{align}
        	    d(f(\V'),\v_1^*)=d(\x',\x)>0,
        	\end{align}
        	
        	a contradiction. With a similar argument we can show that for $\V''=(\v_1,\x)$, $f(\V'')=\x$. Hence, the $\ell_1$-loss for these two preference profiles is:
        	 
        	\begin{align}
        	\ell(\V') &= \absolute{x_1-\frac{x_1}{2}} + \absolute{x_2-\frac{1+x_2}{2}} + \sum_{j=3}^m \absolute{x_j - \frac{x_j}{2}} = 1 - x_2\nonumber, \text{ and} \\
        	\ell(\V'')&= \absolute{x_1-\frac{1+x_1}{2}} + \absolute{x_2-\frac{x_2}{2}} + \sum_{j=3}^m \absolute{x_j - \frac{x_j}{2}} = 1 - x_1\nonumber.
        	\end{align}
        	
        	The optimal mechanism should minimize the quantity $\max\{1-x_1,1-x_2\}$, given that $x_1 + x_2 \leq 1$ and $x_1\geq 0$, $x_2 \geq 0$. Note that $x_1 \leq 1-x_2$, hence  $\max\{1-x_1,1-x_2\} \geq \max\{1-x_1,x_1\}$ which is minimized for $x_1=1/2$ to value at least $1/2$.
        \end{proof}

    	\subsection{A Lower Bound for any Moving Phantom Mechanism}\label{ssec:lowerboundMovingPhantom}
        
        In this subsection we present a preference profile where any phantom mechanism yields loss equal to $1-1/m$. We recall that the {\mechanismname} mechanism achieves this bound for $m=3$. We note that the  counter example construction we present in the proof holds for any even number of voters.
    
        \begin{theorem}\label{thm:moving_phantoms_lower_bound}
        	No moving phantom mechanism can achieve $\ell_1$-loss less than $1-1/m$, for any $m \geq 2$.
        \end{theorem}
        
        \begin{proof}
            Let $n \geq 2$ and even, and let $S=\{1,...,n/2\}$, $Q=\{n/2+1,...,n\}$ be two sets of voters.
            Let $f$ be a moving phantom mechanism defined over $m$ projects and consider the preference profile $\V=(\v_1,...,\v_{n/2},\v_{n/2+1},...,\v_{n})$. All voters $i \in S$ propose the divisions $\v_i=(1,0,...,0)$ while all voters $ k \in Q$  propose the division $\v_k=(1/m,...,1/m)$. Let  $y_i$, for $\in \{0,..,n\}$ denote the $i$-th phantom value and let that $(y_i)_{i \in \{0,...,n\}}$ induce a valid outcome for the moving phantom mechanism. Hence, the outcome of the mechanism is equal to:
        	
        	$$ f_1(\V)=\med\big(\; \overbrace{\tfrac{1}{m},...,\tfrac{1}{m}}^{n/2},y_0,...,y_n,\overbrace{1,...,1}^{n/2} \big)=x$$ while, for $j \in \{2,...,m\}$
        	$$ f_j(\V)=\med\big(\; \overbrace{0,...,0}^{n/2},y_0,...,y_n,\overbrace{\tfrac{1}{m},...,\tfrac{1}{m}}^{n/2} \big)=z.$$ 
        	
        	We will show that both $x \leq 1/m$ and $z \leq 1/m$, which implies that $x=z=1/m$ for the outcome to sum up to $1$, and thus being a valid outcome to moving phantom mechanism. Assume, for the sake of contradiction, that either $x > 1/m$ or $z > 1/m$.
            Starting from the case $z>1/m$, note that $z \leq x$. For the outcome to sum up to $1$, i.e. $x+(m-1)z=1$, it must be $x<1/m$, a contradiction.
            
            We continue with the case $x>1/m$. Then, $n/2$ voters' reports with value equal to $1/m$ should be located in the $n$ lower slots, in the computation of the median for the first project. Hence, at most $n/2$ phantoms can be located in the $n$ lower slots, and the phantom with index $n/2$ should be located in one of the $n+1$ slots higher than the median, i.e. $y_{n/2} \geq x >1/m$. Observe also that $z<1/m$, otherwise $x + (1-m)z >1$. This imply that the $n/2$ voters' reports with value equal to $1/m$ should be located in the $n$ upper slots, in the computation of the medians for the projects $2$ to $m$. Hence, at most $n/2$ phantom values should be located in the $n$ lower slots, and the phantom with index $n/2$ should be located in one of the $n+1$ slots lower than the median, i.e. $y_{n/2} \leq z < 1/m$. A contradiction.
        	
        	Eventually, for a valid outcome of the mechanism it must be that $x=z=1/m$ and the $\ell_1$-loss becomes:
        
        	\begin{align}
        	\ell(\V)&= \absolute{\frac{1}{2}-\frac{1}{m}} + (m-1)\absolute{\frac{1}{2\cdot m}} = 1-\frac{1}{m}.
        	\end{align}
        \end{proof}
        	

        	
        	
        	
        	
        

            \subsection{A Lower Bound for the {Independent Markets} Mechanism}\label{ssec:lowerboundIM}
        In this subsection, we present a class of instances where the \mname{Independent Markets} mechanism from~\cite{freeman2021journal} yields loss at least $0.6862$, for three project and a large enough number of voters $n$. The \mname{Independent Markets} mechanism utilizes the phantoms $(\min\{k \cdot t,1\})_{k \in \zerotok{n}}$.
        
        \begin{theorem}\label{thm:IMlowerbound3}
            The \mname{Independent Markets} mechanism is at least $0.6862$-approximate for three projects.
        \end{theorem}
        
        \begin{proof}
            Let $f$ be the \mname{Independent Markets} mechanism and let $\rho=2-\sqrt{2} \approx 0.5858$. Consider a preference profile $\V$ with $n$ voters, where $\floor{n\rho}$ voters propose the division $(1,0,0)$ while $\ceil{n(1-\rho)}$ voters propose the division $\x=(\sqrt{2}-1,1-\sqrt{2}/2,1-\sqrt{2}/2)$. See also Figure~\ref{fig:matrix_IM3}.
            Let that $t=\frac{\sqrt{2}}{2n}$. Then, $ x_1 =n \rho t \geq \floor{n\rho}\,t$, i.e $\floor{n\rho}+1$ phantom values with indexes $0$ to $\floor{n\rho}$ are at most equal to $x_1$. Hence, there exists $n+1$ values (phantoms and voters' reports) at most equal to $x_1$, thus $f_1(\V)=x_1$.
            Similarly, $x_j= n(1-\rho)\,t \leq \ceil{n(1-\rho)}\,t$ for $j \in \{1,2\}$, i.e. the $n+1 - \floor{n\rho}$ phantom values with indices $\floor{n\rho}$ to $n$ are at least equal to $x_j$. Hence there exists $n+1$ values at least equal to $x_j$, thus $f_j(\V)=x_j$ for $j \in \{2,3\}$.
            
            The $\ell_1$-loss for the preference profile $\V$ is  $$\ell(\V)=\left(3-2\sqrt{2}\right)\left(1 - \frac{\ceil{n(1-\rho)}}{n}\right) +  \frac{\floor{n\rho}}{n}  \geq  0.6862.$$  The inequality holds for $n \geq 2\cdot 10^4$.  
        \end{proof}

\begin{figure}[htb]
    \centering
        \begin{tikzpicture}
            \tikzstyle{column 1}=[anchor=base west]
            \tikzstyle{column 2}=[anchor=base east]
            \tikzstyle{column 3}=[anchor=base]
            \matrix[matrix of math nodes, left delimiter=(, right delimiter=),style={nodes={anchor=base},text width=width("MMMMM"),align=center}] (m)
            {
            1 & 0 & 0  \\
            1 & 0 & 0 &\\
            \vdots & \vdots & \vdots \\
            1 & 0 & 0 &\\
            0.4142 & 0.2929 & 0.2929 \\
            \vdots & \vdots & \vdots \\
            0.4142 & 0.2929 & 0.2929 \\
            };
            
            \node[right=4pt of m-1-3] (right-1) {};
            \node[right=4pt of m-4-3] (right-4) {};
            \node[right=4pt of m-5-3] (right-5) {};
            \node[right=4pt of m-7-3] (right-7) {};

            \node[rectangle,right delimiter=\}] (del-right-1) at ($0.5*(right-1.west) +0.5*(right-4.west)$) {\tikz{\path (right-1.north east) rectangle (right-4.south west);}};
            
            \node[right=22pt] at (del-right-1.west) {$\floor{n\rho}$ voters};
            
            \node[rectangle,right delimiter=\}] (del-right-1) at ($0.5*(right-5.west) +0.5*(right-7.west)$) {\tikz{\path (right-5.north east) rectangle (right-7.south west);}};
            
            \node[right=22pt] at (del-right-1.west) {$\ceil{n(1-\rho)}$ voters};

        \end{tikzpicture}
    \caption{The preference profile which yields a loss at least $0.6862$ for the Independent Markets mechanism.}

    \label{fig:matrix_IM3}
\end{figure}
    \subsection{Lower Bounds for Many Projects}

In this section we provide two impossibility results for large $m$. These results show that the Independent Markets, the {\mechanismname} mechanisms and any utilitarian mechanism may yield loss that approximates $2$, the worst possible $\ell_1$-loss, as $m$ grows. For the two proportional mechanisms, Independent Markets and {\mechanismname}, we use the same construction to show that the $\ell_1$-loss can be as large as $2-\frac{8}{m^{1/3}}$, for large $m$. Then, we focus on the mechanisms maximizing the social welfare and we show an even higher lower bound, at $2-\frac{4}{m+1}$ for every $m \geq 3$.
 
\subsubsection{Proportional Mechanisms}

In this section we show that both the {\Mechanismname} and the \mname{Independent Markets} mechanisms yield an $\ell_1$-loss which is arbitrarily close to $2$, for large enough number of projects. An interesting open question regarding this is whether this holds for any proportional mechanism.

\begin{theorem}\label{thm:LowerBoundLarge_m}
Both the {\Mechanismname} mechanism and the \mname{Independent Markets} mechanism yields loss at least $2- \frac{8}{m^{1/3}} $ for $m \geq 8$.
\end{theorem}

\begin{proof}
We will construct a preference profile with $m$ projects and $m$ voters. In this profile a supermajority of the voters (denoted by the integer variable $z$) are single-minded, towards a unique project each. The rest $m-z$ are fully-satisfied. 

Let $z=\floor{m-m^{2/3}}$ and $a=(m-z)^2$ for some $m \geq 8$.
Let $\s^j$ be the division such that $\s^j_j=1$ for $j \in [m]$ and $\x$ be the division such that $x_{j}=\frac{1}{a+z}$ for $j \in [z]$ and $x_{j}=\frac{m-z}{a+z}$ for $j \in \{z+1,...,m\}$.
Consider an instance $\V$ such that $\v_i=\s^i$ for each $i \in [z]$ while for each $i \in \{z+1,...,m\}$, $\v_i=\x$.

We will first show that the {\mechanismname} mechanism returns the division $\x$ for any $m \geq 8$.
Let  that $t^*=\frac{1}{2}+\frac{1}{2(a+z)}$. Let that $m-z\leq m/2$. This is true for every $m>8$ and implies all phantom values with indices at most $m-z$ are black phantoms, i.e. $y(1,t^*)=\frac{1}{a+z}$. Hence, $f_j(\V)=\frac{1}{a+z}$ for projects $j \in [z]$. To see this, notice in these projects there exists only one voter's report with value strictly higher than $\frac{1}{a+z}$ (the report of the single-minded voter), an well as $m-1$ phantom values (those with indices $2$ to $m$). Since $\frac{1}{a+z}$ is the largest of the other $m+1$ values. Similarly, $y(m-z,t^*)=\frac{m-z}{a+z}$ and  $f_j(\V)=\frac{m-z}{a+z}$ for projects $j \in \{z+1,...,m\}$; there exists $z$ voters' reports with value $0$ and $m-z+1$ phantom values smaller than $y(m-k,t^*)$, while the smaller value of the rest is equal to $\frac{m-z}{a+z}$.
In a similar manner, by using $t^*=\frac{1}{a+z}$, the \mname{Independent Markets} mechanism returns the same outcome.

To compute the loss for $\V$ for the outcome $\x$, first notice that $\absolute{\bar{v}_j-f_j}= \frac{1}{m} + \frac{m-z}{m}\frac{1}{a+z} - \frac{1}{a+z} = \frac{a}{m(a+z)}$ for all $j \in [z]$ and $\absolute{\bar{v}_j-f_j}= \frac{m-z}{a+z} -\frac{m-z}{m} \cdot \frac{m-z}{a+z}=\frac{(m-z)z}{m(a+z)}$. 
(Note that $\frac{1}{m} + \frac{m-z}{m}\frac{1}{a+z} \geq \frac{1}{a+z} $ and $\frac{m-z}{a+z} \leq \frac{m-z}{m} \cdot \frac{m-z}{a+z} $). Hence the loss is equal to

\begin{align}
\ell(\V)&= z\cdot \frac{a}{m(a+z)} + (m-z)\cdot \frac{(m-z)z}{m(a+z)} \nonumber \\
& = 2\cdot\frac{a}{a+z} \cdot\frac{z}{m} \nonumber \\
& > 2 \left( \frac{m^{4/3}}{m^{4/3}+2m}\right)\left(\frac{m-m^{2/3}-1}{m}\right) \nonumber \\
& \geq 2 \left( 1- \frac{2}{m^{1/3}}\right)^2 \nonumber \\
& \geq  2- \frac{8}{m^{1/3}} 
\end{align}

For the first inequality we have used the following facts: $z > m-m^{2/3}-1$, $ a \geq m^{4/3}$ and $ a+z \leq m^{4/3}+2m$ for $m \geq 5$. The first one is due to $ x \geq \floor{x}>x-1$ for any $x \in \mathbb{R}$.
This implies alos that $a = (m - \floor{m-m^{2/3}})^2 \geq m^{4/3}$. Also, notice that $a+z = (m - \floor{m-m^{2/3}})^2 +  \floor{m-m^{2/3}} \leq (m^{2/3}+1)^2 + m-m^{2/3} +1 \leq m^{4/3}+2m$.
Second inequality is due to $\left( \frac{m^{4/3}}{m^{4/3}+2m}\right) \geq 1-\frac{2}{m^{1/3}}$ for $m \geq 0$ and 
$\left(\frac{m-m^{2/3}-1}{m}\right) \geq 1-\frac{2}{m^{1/3}}$ for $m \geq 1$. The last inequality is due to $\left(1-\frac{c}{x}\right)^2 \geq \left(1-\frac{2c}{x}\right)$ for $c\geq 0$ and $x>0$.
\end{proof}


\begin{figure}[htb]
    \centering
        \begin{tikzpicture}
            \tikzstyle{column 1}=[anchor=base west]
            \tikzstyle{column 2}=[anchor=base east]
            \tikzstyle{column 3}=[anchor=base]
            \matrix[matrix of math nodes, left delimiter=(, right delimiter=),style={nodes={anchor=base},text width=width("MM"),align=center}] (m)
            {
            1 & 0 & 0 &\cdots & 0 & 0 & \cdots & 0 \\
            0 & 1 & 0 &\cdots & 0 & 0 & \cdots & 0 \\
            0 & 0 & 1 &\cdots & 0 & 0 & \cdots & 0 \\
            \vdots & \vdots & \vdots &\ddots & \vdots & 0 & \cdots & 0\\
            0 & 0 & 0 &\cdots & 1 & 0 & \cdots & 0 \\
            \frac{1}{a+z} & \frac{1}{a+z} & \frac{1}{a+z} &\cdots & \frac{1}{a+z} & \frac{m-z}{a+z} & \cdots & \frac{m-z}{a+z} \\
            \vdots & \vdots & \vdots &\vdots & \vdots & \vdots & \cdots & \vdots\\
            \frac{1}{a+z} & \frac{1}{a+z} & \frac{1}{a+z} &\cdots & \frac{1}{a+z} & \frac{m-z}{a+z} & \cdots & \frac{m-z}{a+z} \\
            };
            \draw[dashed] ($0.5*(m-1-5.north east)+0.5*(m-1-6.north west)$) -- ($0.5*(m-8-6.south east)+0.5*(m-8-5.south west)$);
            \draw[dashed] ($0.5*(m-5-1.south west)+0.5*(m-6-1.north west)$) -- ($0.5*(m-5-8.south east)+0.5*(m-6-8.north east)$);
            
            \node[above=1pt of m-1-1] (top-1) {};
            \node[above=1pt of m-1-5] (top-5) {};
            \node[above=-1pt of m-1-6] (top-6) {};
            \node[above=-1pt of m-1-8] (top-8) {};
            
            \node[right=4pt of m-1-8] (right-1) {};
            \node[right=4pt of m-5-8] (right-5) {};
            \node[right=4pt of m-6-8] (right-6) {};
            \node[right=4pt of m-8-8] (right-8) {};
            
            \node[rectangle,above delimiter=\{] (del-top-1) at ($0.5*(top-1.south) +0.5*(top-5.south)$) {\tikz{\path (top-1.south west) rectangle (top-5.north east);}};
            \node[above=10pt] at (del-top-1.north) {$z$};
            \node[rectangle,above delimiter=\{] (del-top-2) at ($0.5*(top-6.south) +0.5*(top-8.south)$) {\tikz{\path (top-6.south west) rectangle (top-8.north east);}};
            \node[above=10pt] at (del-top-2.north) {$m-z$};

            \node[rectangle,right delimiter=\}] (del-right-1) at ($0.5*(right-1.west) +0.5*(right-5.west)$) {\tikz{\path (right-1.north east) rectangle (right-5.south west);}};
            \node[right=22pt] at (del-right-1.west) {$z$};
            \node[rectangle,right delimiter=\}] (del-right-2) at ($0.5*(right-6.west) +0.5*(right-8.west)$) {\tikz{\path (right-6.north east) rectangle (right-8.south west);}};
            \node[right=22pt] at (del-right-2.west) {$m-z$};

        \end{tikzpicture}
    \caption{The preference profile which yields a loss at least $2- \frac{8}{m^{1/3}}$ for both the {\mechanismname} and the Independent Markets mechanisms.}

    \label{fig:matrix_large_m}
\end{figure}

\subsubsection{Utilitarian mechanisms}

In this part, we focus on a family of budget aggregation mechanisms, which we call \emph{utilitarian} mechanisms. A utilitarian mechanism returns an aggregated division which maximizes the social welfare. This is done by minimizes the $\ell_1$ distance between the outcome and each voter's proposal, i.e. for a utilitarian mechanism $f$ and a preference profile $\V=(\v_1,...,\v_n)$, then

\begin{equation}
	f(\V) \in \argmin_{\x \in \mathcal{D}(m)} \sum_{i \in [n]} d( \v_i, \x ).
\end{equation}

Mechanisms that maximize the social welfare have been examined in the literature (see Section~\ref{sec:Intro}). One of them is a moving phantom mechanism and is proven to be the unique Pareto Optimal mechanism in this family, providing a dichotomy between Proportional the Pareto Optimal moving phantom mechanisms. Here we show that any utilitarian mechanism, inevitably yields very high loss. In contrast to the two proportional mechanisms, these mechanisms yields quite large loss even for small $m$, e.g. for three project, the $\ell_1$-loss is at least $1$.   

\begin{theorem}\label{thm:lower_bound_utilitarian}
	For any utilitarian mechanism, there exists an instance with $\ell_1$-loss equal to $2-\frac{4}{m+1}$.
\end{theorem}

\begin{proof}
Let $\s^j$ be the division such that $s^j_j=1$ for $j \in [m]$.
Consider a preference profile with $m$ projects and $m+1$ voters. Each voter $i \in [m+1]$ is single-minded, and $\v_i=\s^i$ when $i \mod m =0$. Hence, project $1$ its fully supported by two voters, while the other projects are supported by one voter each.

Let $f$ be a utilitarian mechanism and let $f(\V)=\x$.
Let $x_1=1-\epsilon$ and $\sum_{j=2}^m x_j=\epsilon$ for some $\epsilon \in [0,1]$. We call the quantity $\sum_{i \in [n]} d( \v_i, \x )$ the \emph{social cost} (SC). Recall that a utilitarian mechanism minimizes the social cost.

\begin{align}
	\text{SC} &= \sum_{i =1 }^{m+1} d( \v_i, \x )  \nonumber \\
	&=  \sum_{i=1 }^{m+1} \sum_{j=1}^{m} \absolute{v_{i,j} - x_j}  \nonumber \\
	& = 2(1-x_1) + 2\sum_{j=2}^m x_j +  \sum_{i =2 }^m \left(  2-2x_i \right) \label{eq:SWlb1} \nonumber \\
	&= 2\epsilon + 2(m-1) 
\end{align}

The social cost is minimized for $\epsilon=0$, i.e. the only possible outcome of a utilitarian mechanism for this preference profile is $x_1=1$ and $x_j$ for $j \in \{2,...,m\}$ for $j \in [m]$.

On the other hand, the proportional division assigns $\bar{\V}_1=\frac{2}{m+1} $ and $\bar{\V}_j=\frac{1}{n+1}$ for any $j \in \{2,...,m\}$. Eventually, the $\ell_1$ loss is:
\begin{align}
\ell(\V) &= \left(1-\frac{2}{m+1}\right) + \sum_{j=2}^m \frac{1}{m+1} \nonumber \\
&= 2-\frac{4}{m+1}.
\end{align} 
\end{proof}

\begin{figure}[htb]
    \centering
    \begin{tikzpicture}
    \matrix[matrix of math nodes, left delimiter=(, right delimiter=),style={nodes={anchor=base},text width=width("MM"),align=center}]
            {
            1 & 0 & 0 & 0 \\
            0 & 1 & 0 & 0 \\
            0 & 0 & 1 & 0 \\
            0 & 0 & 0 & 1 \\
            1 & 0 & 0 & 0 \\
            };
    \end{tikzpicture}
    \caption{An example of the construction used in Theorem \ref{thm:lower_bound_utilitarian},  with $5$ voters and $4$ projects.}
    \label{fig:matrix_utilitarian}
\end{figure}

        \section{Discussion}
    
    This paper proposes an approximation framework that rates budget
    aggregation mechanisms according to the worst-case distance from
    the proportional allocation, a natural fairness desideratum. We
    propose essentially optimal mechanisms within the class of moving phantom
    mechanisms for the cases of two and three projects. The most
    interesting open question is whether there exists any
    $(2-\epsilon)$-approximate mechanism, for some constant
    $\epsilon>0$, with an arbitrary number of projects. Our constructions in Section~\ref{sec:lowerbounds} show that our mechanism, and two mechanisms already explored in the literature cannot yield a bound asymptotically better than $2$. 
    
    
    Moving beyond moving phantom mechanisms, one can ask the questions: Is there any mechanism with worst-case $\ell_1$-loss smaller than $2/3$ for the case of three projects? Our lower bound for any truthful mechanism is just $1/2$. While this lower bound is very simple, and probably a more sophisticated construction may answer the question negatively, we should note that until now the existence of one such mechanism is still possible. Even more, one could ask for mechanisms with improved approximation guarantees (e.g. even below $1/2$ by using a relaxed version of truthfulness).
    
    Our approximation framework can be expanded for other desirable properties. For example, it is natural to study moving phantom mechanisms with good approximation guarantees with respect to the \emph{egalitarian} or the \emph{Nash} social welfare (see~\cite{aziz2021participatory}).
    We also note that similar notions of approximation can be defined for other participatory budgeting problems. Consider for example the case where the voters divide the budget using approval voting, and each voter's utility is the proportion of the budget given to the projects she approves (see~\cite{bogomolnaia2005} for this line of work). A well-known fairness notion under this model is the \emph{fair share}, which demands that each voter has a utility of at least $1/n$, where $n$ is the number of voters. For this example one can use as approximation the worst-case distance between the outcome of the mechanism and all the divisions which satisfy the fair share property. 
    

    \bibliographystyle{abbrvnat}
    \bibliography{references}
    \appendix
    \appendix
\section{The {\mechanismname} mechanism is a Moving Phantom mechanism}\label{appdx:Inclusion}

In this section we show that the {\mechanismname} mechanism is a moving phantom mechanism. For that, we will present an alternative phantom system $\mathcal{Y}^{\text{PU}'}$ which satisfies Definitions~\ref{def:moving_phantoms}. Then, we show that the {\mechanismname} mechanism simulates this new alternative definition.

\noindent The alternative phantom system is $\mathcal{Y}^{\text{PU}'}=\{ y'_k(t): k \in  \zerotok{n} \}$, for which 

            \begin{align}\label{eq:NewMechanism:G:1}
                y'_k(t)=\begin{cases}
                0 & \frac{k}{n} < \frac{1}{2} \\
                {\frac {4tk}{n\left( \frac{1}{2}-\epsilon\right)}}-2\,\frac{t}{\frac{1}{2}-\epsilon} & \frac{k}{n} >\geq \frac{1}{2} \\
                \end{cases}
            \end{align}
            for $t<1/2 - \epsilon$,
            
            \begin{align}\label{eq:NewMechanism:G:2}
                y'_k(t)=\begin{cases}
                    \frac{k(2\,t-1)}{n} + \frac{2k\epsilon}{n} + \frac{2k\epsilon}{n}  &  \frac{k}{n} < \frac{1}{2} \\
                    \frac{k(3-2\,t)}{n} -2 + 2\,t + 2\epsilon - \frac{2k\epsilon}{n} &  \frac{k}{n} \geq \frac{1}{2}, \\
                \end{cases}
            \end{align}
            
        \noindent    for $ 1/2 - \epsilon \leq t<1-\epsilon$, and
            
            \begin{align}
                y'_k(t)= \frac{k}{n}\left( \frac{1-t}{\epsilon}\right) + \frac{t-1}{\epsilon}+1,
            \end{align}
            
         \noindent   for $t>1-\epsilon$, for some $0<\epsilon<1/2$.

         Observe that this alternative phantom system satisfies Definition~\ref{def:moving_phantoms}: All phantom functions are continuous, $y_{k}(0)=0$ and $y_{k}(1)=1$ for all $k \in \zerotok{n}$. Also, $y_{k+1}(t) \geq y_k(t)$ for all $k \in \zerotok{n-1}$.
        
        We notice also that there exists no feasible solution for $t>1-\epsilon$, for this phantom system. For $t=1-\epsilon$ the phantoms of the mechanism described by the phantom system $\mathcal{Y}^{\text{PU}'}$ are exactly the phantoms used by the Uniform Phantom mechanism. In the following lemma we show that the sum of the medians of the Uniform Phantom mechanism is at least $1$. Hence, any phantom returned by the phantom system $\mathcal{Y}^{\text{PU}'}$ with $t>1-\epsilon$, would return an outcome that sums to a value strictly larger than $1$.
        
    \begin{lemma}\label{lem:UniformSumsGreaterThan1}
	    Let $\x$ be the outcome of the Uniform Phantom mechanism on an arbitrary preference profile for some $m\geq 3$. Then $\sum_{j \in [m]} x_j \geq 1$.
    \end{lemma}

    \begin{proof}
	Let $k_j$ be the largest index such that $\frac{k_j}{n}\leq x_j$.
	Assume for the sake of contradiction that $\sum_{j \in [m]} x_j<1$, i.e. $\sum_{j \in [m]} \frac{k_j}{n}<1$. In the slots $1$ to $n+1$, there exist exactly $k_j$ phantom values, for each project $j \in [m]$. As a result, there exist exactly $n-k_j$ voters' reports in the same slots. In total, $mn-\sum_{j \in [m]}{k_j}>n(m-1)$ voters' reports are located in the slots $1$ to $n+1$. Similarly, there exist exactly $k_j$ voters' reports in each project $j$ in the upper slots (slot $n+2$ to slot $2n+1$). Hence in total there exist exactly $\sum_{j \in [m]} k_j < n $ upper slots filled by phantom values, out of $mn$ slots. Since there are $mn$ slots in the upper phantoms, there should be at least $n(m-1)+1$ voters' reports, but we already know that at least $n(m-1)$ out of $nm$ voters' reports are located either in the lower slots or in the slots of the medians. A contradiction.
    \end{proof}

    We are ready now to show that the two phantom systems $\mathcal{Y}^{\text{PU}}$ and $\mathcal{Y}^{\text{PU}'}$ describe the same moving phantom mechanism.
    
    \begin{lemma}
        The phantom systems $\mathcal{Y}^{\text{PU}}$ and $\mathcal{Y}^{\text{PU}'}$ implement the same moving phantom mechanism. 
    \end{lemma}
    
    \begin{proof}
    We use $y_k(t)$ to denote the functions from the phantom system $\mathcal{Y}^{\text{PU}}$ and $y'_k(t)$ to denote the functions for the phantom system $\mathcal{Y}^{\text{PU}'}$. Consider any preference profile $\V$ over $m$ projects. Let $f_j(\V)=\median{\V_{i \in [n],j},(y_k(t))_{k \in \zerotok{n}}}$ and $f'_j(\V)=\median{\V_{i \in [n],j},(y'_k(t'))_{k \in \zerotok{n}}}$, for suitable $t \in [0,1]$ and $t' \in [0,1]$ such that $\sum_{j \in m} f_j(\V)=1$ and $\sum_{j \in [m]} f'_j(\V)=1$. We will show that the {\mechanismname} implements the mechanism described by $\mathcal{Y}^{\text{PU}'}$.

    We consider the phantom system $\mathcal{Y}^{\text{PU}'}$. Assume that $t'>1-\epsilon$. By Lemma \ref{lem:UniformSumsGreaterThan1},  $\sum_{j \in [m]} f'_j(\V)>1$, i.e. no feasible outcome is possible with $t'>1-\epsilon$. Assume now that $1/2 -\epsilon < t' \leq 1-\epsilon$. Then, by using $t=t'+\epsilon$, the tuples $(y'_k(t'))_{k \in \zerotok{n}}$ and $(y_k(t))_{k \in \zerotok{n}}$ are equivalent, hence $f'(\V)=f(\V)$. Finally, assume that $f(\V)$ uses some $t\leq 1/2 - \epsilon$ and returns a feasible solution. Then, by using $t'= \frac{t}{1-2\epsilon}$ , the tuples $(y'_k(t'))_{k \in \zerotok{n}}$ and $(y_k(t))_{k \in \zerotok{n}}$ are equivalent, hence $f'(\V)=f(\V)$.
    
     
    \end{proof}
        

\section{Completeness of Theorem~\ref{thm:PUmechanism_upper_bound}}\label{appdx:zeros}

In this section we tackle the remaining case for the proof of Theorem~\ref{thm:PUmechanism_upper_bound}, where no zero values exists in the outcome. Our technique is similar to Section~\ref{ssec:upperbound3}.

Due to Theorem~\ref{thm:WorstCaseInstances} we can focus on {\profilename} profiles to upper bound the $\ell_1$-loss. Recall that the division $\x=(x_1,x_2,x_3)$ represents the outcome of the mechanism on a three-type profile. First, note that when the outcome includes two $0$ values, i.e. $\x=(1,0,0)$, any {\profilename} profile $\V$ contains only single-minded voters. Since the mechanism is proportional, the loss is $0$, and the mehanism is optimal for such profiles. We turn now our attention to the case where there exists only one $0$ value in the outcome, say $x_3=0$. For that, we need to check two possibilities. If $t > 1/2$, there can be only one phantom value equal to $0$, by the definition of the mechanism, and for $x_3=0$, at least $n$ voters' reports should be equal to $0$. Thus all voters propose $0$ for project $3$. This can be reduced to the case of $2$ projects. The {\mechanismname} mechanism can ensure a feasible solution by using the phantoms $(k/n)_{k \in \zerotok{n}}$, which we can enforce by setting $t^*=1$, i.e. the mechanism simulates the Uniform Phantom mechanism for this case. By Theorem~\ref{thm:Uniform2projects}, the $\ell_1$-loss in this case cannot be higher than $1/2$.

We are left with a single case: let that $x_1>0$, $x_2>0, x_3=0$ and $t \leq 1/2$. We will build a Non Linear Program to tackle this case. We first recall that $b_{1,2}$ and $b_{2,1}$ counts divisions $(x_1,x_2,0)$, $b_{1,3}$ counts divisions $(x_1,0,x_2)$, $b_{2,1}$ counts divisions $(x_1,x_2,0)$, $b_{2,3}$ counts divisions $(0,x_2,x_1)$, $b_{3,1}$ counts divisions $(1,0,0)$ and $b_{3,2}$ counts divisions $(0,1,0)$. Notice that for projects $j \in \{1,2\}$ the only possible voters' reports are $0,x_j$ and $1$. Hence, for some $t^*$ which enforces a valid outcome it must hold that $y(a_1 + b_{3,1},t^*) \leq x_1 \leq y(1-a_2-a_3-b_{2,3}-b_{3,2},t^*)$ and $y(a_2 + b_{3,2},t^*) \leq x_1 \leq y(1-a_2-a_3-b_{2,3}-b_{3,2},t^*)$. 

We know from the above paragraph that for project $3$, all complementary values $1-x_1$ and $1-x_2$ are positive, i.e. they are located in the upper slots. This is not always the case for projects $1$ and $2$, where we can only guarantee that $a_1$ and $a_2$ voters' reports (i.e. the $1$ valued reports) are located in the upper slots, respectively. On the other side, there exists $C+b_{2,3}+b_{3,2}+a_2+a_3$ zero values in project $1$ and $C+b_{1,3}+b_{3,1}+a_1+a_3$ from the double-minded, single-minded and the zeros of the fully-satisfied voters. In the following lemma, we use this information to impose sufficient and necessary conditions (similar to Lemma~\ref{lem:NecessaryConditions}) for the mechanism to return $\x$ as an outcome. 

    \begin{lemma}\label{lem:NC_t<1/2}
        Let that $x_1>0$, $x_2>0$, $x_3=0$. For any moving phantom mechanism $f$, defined by the phantom system $\mathcal{Y}=\{y_k(t): k \in \zerotok{n} \}$, and a {\profilename} profile $\V$ then $f(\V)=\x$ if and only if:  
        \begin{align}
            y_{a_1}(t^*) &\leq  x_1 \leq y_{ n-a_2 -a_3 +b_{2,3} + b_{3,2} }(t^*)  \label{eq:NC_t<1/2:1} \\
            y_{a_2}(t^*) &\leq  x_2 \leq y_{ n-a_1 -a_3 +b_{1,3} + b_{3,1} }(t^*)  \label{eq:NC_t<1/2:2}
        \end{align}
        
        for any \[t^* \in \left\{ t: \sum_{j \in [m] } \median{ \V_{i \in [n],j},(y_k(t))_{k \in \zerotok{n}} } =1 \right\}.\]
    \end{lemma}

    \begin{proof}
        We focus on inequality \ref{eq:NC_t<1/2:1}. Similar arguments can be used to show inequality \ref{eq:NC_t<1/2:2}.

    (if direction)
    Let $\V$ be {\profilename} profile and $f(\V)=\x$ for some $t^* \in [0,1]$. Assume for the sake of contradiction that $y_{a_1}(t^*)>x_1$. This implies that the $n-a_1+1$ phantoms with indices $a_1,...,n$ are located in the upper slots. Since the $a_1$ voters' reports should be located in the upper slots, at least $n+1$ values should be located in the lower slots. A contradiction.
    Suppose now, that $y_{ n-a_2 -a_3 +b_{2,3} + b_{3,2}}(t^*)<f_1(\V)$. Then, $n -a_2 -a_3 +b_{2,3} + b_{3,2}+1$ phantoms are located in the lower slots. Also $a_3+a_2+b_{2,3}+b_{3,2}$ zero values should be located in the lower slots. These are are $ n +1$ values, while the lower slots are $n$. A contradiction. 
    
    (only if direction)
    Let that inequalities \ref{eq:NC_t<1/2:1} hold, and consider a {\profilename} profile $\V$, where $f_1(\V)<x_1$. Hence the $C$ reports for the fully satisfied voters, the  $b_{1,2}+b_{1,3}+b_{2,1}$ reports equal to $x_1$ for the double-minded voters and the $a_1+b_{3,1}$ $1$-valued reports, should be located in the upper slots. From inequality \ref{eq:NC_t<1/2:1}, there exists $a_2+a_3+b_{2,3}+b_{3,2}+1$ phantoms in the upper slots. These are $n+1$ values, which cannot fit in the $n$ upper slots. Similarly, let that $x_1<f_1(\V)$. Then $C+b_{1,2}+b_{1,3}+b_{2,1}+a_1+b_{3,1}+a_2+b_{2,3}+b_{3,2}$ voters' reports are located in the lower slots. From inequality \ref{eq:NC_t<1/2:1}, at  least $a_1+1$ voter reports are located in the lower slots. Hence, $n+1$ values should be located in the lower slots. A contradiction.
    \end{proof}

\begin{table}[tb]
    \centering
    \begin{tabular}{|c|c|c|c|c|c|}
        \hline
         \textbf{signs} & \textbf{phantoms} & \textbf{Status}& \textbf{ Loss (lower bound)} &  \textbf{ Loss (upper bound)} & \textbf{Gap} \\ \hline
            $(-,+,+)$ & $(r,r),(r,r)$ & INFEASIBLE & $-$ & $-$ & $-$ \\
            $(-,+,+)$ & $(b,r),(r,r)$ & OPTIMAL & $0.000000$ & $0.000000$ & $\num{0.00e+00}$ \\ 
            $(-,+,+)$ & $(r,r),(b,r)$ & OPTIMAL & $0.000000$ & $0.000000$ & $\num{0.00e+00}$ \\ 
            $(-,+,+)$ & $(b,r),(b,r)$ & OPTIMAL & $0.250002$ & $0.250009$ & $\num{7.49e-06}$ \\
            $(-,-,+)$ & $(r,r),(r,r)$ & INFEASIBLE & $-$ & $-$ & $-$ 
            \\
            $(-,-,+)$ & $(b,r),(r,r)$ & OPTIMAL & $0.000000$ & $0.000000$ & $\num{0.00e+00}$ \\ 
            $(-,-,+)$ & $(r,r),(b,r)$ & OPTIMAL & $0.000000$ & $0.000000$ & $\num{0.00e+00}$ \\ 
            $(-,-,+)$ & $(b,r),(b,r)$ & OPTIMAL & $0.499999$ & $0.500003$ & $\num{3.43e-06}$ \\ 
            \hline
    \end{tabular}
    \caption{The upper bounds computed by the QPQCs, for $t<1/2$ and $x_3=0$. The lower bound corresponds to the largest loss for a feasible solution computed by the solver. The upper bound corresponds to the smaller non-feasible lower bound computed be the solver. The last column shows the gap between them. Gaps smaller than $10^{-5}$ are insignificant due to the tolerance of the solver.}
    \label{tab:results:special}
\end{table}

Using Lemma~\ref{lem:NC_t<1/2}, we build a NLP to upper bound the $\ell_1$-loss for this case, described in detail in Figure~\ref{fig:NLP:worst-case:zeroinoutcome}. Following our techniques from Section~\ref{ssec:upperbound3}, we solve this NLP using a set of simpler, Quadratic Programs with Quadratic Constraints. 

Since $x_3=0$, note that $\bar{v}_3 -x_3 \geq 0$ and the only sign patterns we need to tackle are $(-,-,+)$ and $(-,+,+)$. Recall that $(+,+,+)$ is only possible when the $\ell_1$-loss is equal to $0$. 

For the phantom patterns, note (using Lemma~\ref{lem:NC_t<1/2}) that $n-a_2-a_3 - b_{2,3} - b_{3,2} >1/2$ and $n-a_1-a_3-b_{1,3}-b_{3,1}>1/2$; otherwise at least one of the $x_1$ or $x_2$ is equal to $0$, a case we have already covered. Hence, the upper bounds in inequalities~\ref{eq:NC_t<1/2:1} and~\ref{eq:NC_t<1/2:2} refer to red phantoms. We don't have any guarantee for the lower phantoms, hence we use a phantom pattern in $\{(b,r),(r,r)\}^2$ for these cases. We check all $4$ possible sign patterns, for projects $1$ and $2$. We don't need to examine project $3$, since $x_3 \leq y_{k}(t)$ for any $k \in \zerotok{n}$ and any $t \in [0,1/2]$. Eventually, we need to check $2 \times 4$ QPQCs. By using the solver, we provide upper bounds for each program, which do not exceed $1/2+\epsilon$, for some $\epsilon$ no larger than $10^{-5}$. The detailed computed upper bounds are depicted in Table~\ref{tab:results:special}. 

\begin{figure}[t]
	\begin{align}
	\text{maximize }&\quad \sum_{j=1}^3 \absolute{ \bar{v}_j-x_j} \label{MIQCP:objective:worst-case:zero} \\
	\text{subject to}& && \nonumber \\ 
	\sum_{j=1}^3 x_j& =1 && \nonumber  \\
	\hat{A} &=\sum_{j=1}^3 \hat{a}_j && \nonumber \\
	\hat{B} &= \sum_{j,k \in [3],j\neq k} \hat{b}_{k,j} && \nonumber \\
	\hat{y}_{\hat{a}_1}(t^*) &\leq  x_1 \leq \hat{y}_{1-\hat{a}_2-\hat{a}_3-\hat{b}_{2,3}-\hat{b}_{3,2}}(t^*)  \nonumber \\
    \hat{y}_{\hat{a}_2}(t^*) &\leq  x_2 \leq \hat{y}_{1-\hat{a}_1-\hat{a}_3-\hat{b}_{1,3}-\hat{b}_{3,1}}(t^*) \nonumber \\
    x_3 &=0, \nonumber \\
    \hat{A}+\hat{B} & \leq 1 && \nonumber \\
	 x_j &\geq 0, \hat{a}_j \geq 0 ,&& \forall j \in [3] \nonumber \\
	 \hat{b}_{k,j} &\geq 0 ,&& \forall j,k \in [3]\nonumber \\
	 0 &\leq t^* \leq 1/2. \nonumber
	\end{align}
	\caption{The Non-Linear Program used to upper bound the maximum $\ell_1$-loss for the {\mechanismname} mechanism for the special case $x_3=0$ and $t<1/2$.}
	\label{fig:NLP:worst-case:zeroinoutcome}
	\end{figure}

\section{Examples of QPQCs}

In this section we provide some detailed examples of the Quadratic Programs with Quadratic Constraints we have used for the proof of Theorem~\ref{thm:PUmechanism_upper_bound}. Figure~\ref{fig:NLP:worst-case} depicts an example of the first $27$ programs Quadratic Programs we solve to get an initial upper bound, slightly higher than $2/3$. Figure~\ref{fig:NLP:++-} depicts one of the $81$ Quadratic Programs we solve to to check if any other case may yield loss higher than $2/3$. Finally, in Figure~\ref{fig:NLP:worst-case:special_examples} we present a specific Quadratic Program for the special case where there exists a single zero value in the outcome.

	\begin{figure}[t]
	\begin{align}
	\text{maximize }&\quad \bar{v}_1-x_1 + \sum_{j=2}^3 { x_j-\bar{v}_j}, \nonumber \\
	\text{subject to}& ,&& \nonumber \\ 
	\sum_{j=1}^3 x_j& =1 ,&& \nonumber  \\
	\hat{A} &=\sum_{j=1}^3 \hat{a}_j ,&& \nonumber \\
	\hat{B} &= \sum_{j,k \in [3],j\neq k} \hat{b}_{k,j} ,&& \nonumber \\
	\hat{z}_j&= \hat{a}_j + \sum_{k \in [3]\setminus\{j\}}{\hat{b}_{k,j}} ,&& \forall j \in [3] \nonumber \\
	\hat{q}_j&= \sum_{k \in [3]\setminus\{j\}}{\hat{b}_{j,k}} ,&& \forall j \in [3] \nonumber \\
	 x_1 & \geq \hat{z}_1  (3-2t^*)-2+2t^*, \nonumber \\
	 \hat{z}_1 &\geq 1/2, \nonumber \\
	 1/2 &\leq \hat{C} + \hat{z}_1 + \hat{q}_1, \nonumber \\
	 x_1 &\leq \left(\hat{C} + \hat{z}_1  + \hat{q}_2\right)(3-2t^*)-2+2t^*, \nonumber   \\
	 x_j & \geq \hat{z}_j (2t^*-1), && \forall j \in \{2,3\}  \nonumber  \\
	 x_j &\leq \left(\hat{C} + \hat{z}_j + \hat{q}_j  \right)(2t^*-1),&&  \forall j \in \{2,3\} \nonumber   \\
	 \hat{z}_j &\leq 1/2 ,&& \forall j \in \{2,3\} \nonumber \\
	 1/2 &\geq \hat{C} +\hat{z}_j + \hat{q}_j  ,&& \forall j \in \{2,3\} \nonumber \\
     \hat{A}+\hat{B} & \leq 1 ,&& \nonumber \\
	 x_j &\geq 0, \hat{a}_j \geq 0 ,&& \forall j \in [3] \nonumber \\
	 \hat{b}_{k,j} &\geq 0 ,&& \forall j,k \in [3] \nonumber \\
	 1/2 &\leq t^* \leq 1. \nonumber
	\end{align}
	\caption{The Quadratic Program with Quadratic Constraints for the maximum loss computation for the case $t>1/2$, $(+,-,-)$, $((r,r),(b,b),(b,b))$. The inequalities referring to function $\hat{y}$ in Figure~\ref{fig:NLP} are now replaced by $8$ new inequalities, for the specified phantom pattern. }
	 \label{fig:NLP:worst-case}
	\end{figure}

	\begin{figure}[t]
	\begin{align}
	\text{maximize }&\quad 1 \nonumber \\
	\text{subject to}& && \nonumber \\
	2/3 &\leq \bar{v}_1-x_1 + \sum_{j=2}^3 { x_j-\bar{v}_j}, \nonumber \\
	\sum_{j=1}^3 x_j& =1, && \nonumber  \\
	\hat{A} &=\sum_{j=1}^3 \hat{a}_j, && \nonumber \\
	\hat{B} &= \sum_{j,k \in [3],j\neq k} \hat{b}_{k,j}, && \nonumber \\
	\hat{z}_j&= \hat{a}_j + \sum_{k \in [3]\setminus\{j\}}{\hat{b}_{k,j}}, && \forall j \in [3] \nonumber \\
	\hat{q}_j&= \sum_{k \in [3]\setminus\{j\}}{\hat{b}_{j,k}}, && \forall j \in [3] \nonumber \\
	 x_j & \geq \hat{z}_j(2t^*-1), && \forall j \in [3] \nonumber \\
	 x_j &\leq \left(\hat{C} + \hat{z}_j+ \hat{q}_j \right)(2t^*-1),&&  \forall j \in [3]  \nonumber \\
	 \hat{z}_j &\leq 1/2 ,&&  \forall j \in [3] \nonumber \\
	 \hat{C} + \hat{z}_j + \hat{q}_j&\leq 1/2 ,&&  \forall j \in [3] \nonumber \\
     \hat{A}+\hat{B} & \leq 1, && \nonumber \\
	 x_j &\geq 0, \hat{a}_j \geq 0 ,&& \forall j \in [3] \nonumber \\
	 \hat{b}_{k,j} &\geq 0 ,&& \forall j,k \in [3] \nonumber \\
	 0 &\leq t^* \leq 1/2. \nonumber
	\end{align}
	\caption{The Quadratic Program with Quadratic Constraints which checks whether the case $t\leq 1/2$, $(+,+,-)$, $((b,b),(b,b),(b,b))$ yields $\ell_1$-loss greater than $2/3$. The inequalities referring to function $\hat{y}$ in Figure~\ref{fig:NLP} are now replaced by $4$ new inequalities, for the specified phantom pattern. This program has no feasible solutions.}
	 \label{fig:NLP:++-}
	\end{figure}

\begin{figure}
	\begin{align}
	\text{maximize }&\quad \sum_{j=1}^2 { x_j-\bar{v}_j} + \bar{v}_3  \nonumber\label{MIQCP:objective:worst-case:special} \\
	\text{subject to}& && \nonumber \\ 
	\sum_{j=1}^2 x_j& =1, && \nonumber  \\
	\hat{A} &=\sum_{j=1}^3 \hat{a}_j, && \nonumber \\
	\hat{B} &= \sum_{j,k \in [3],j\neq k} \hat{b}_{k,j}, && \nonumber \\
    x_1 &\leq ({1-\hat{a}_2-\hat{a}_3-\hat{b}_{2,3}-\hat{b}_{3,2}})\,\frac{4t^*}{n} -2\,t^*,  \nonumber \\
    x_2 &\leq ({1-\hat{a}_1-\hat{a}_3-\hat{b}_{1,3}-\hat{b}_{3,1}})\,\frac{4t^*}{n} -2\,t^*, \nonumber \\
    x_3 &=0, \nonumber \\
    \hat{a}_1 &\leq 1/2, \nonumber \\
    \hat{a}_2 &\leq 1/2, \nonumber \\
    1/2 &\leq {1-\hat{a}_2-\hat{a}_3-\hat{b}_{2,3}-\hat{b}_{3,2}}, \nonumber \\
    1/2 &\leq {1-\hat{a}_1-\hat{a}_3-\hat{b}_{1,3}-\hat{b}_{3,1}}, \nonumber \\
    \hat{A}+\hat{B} & \leq 1, && \nonumber \\
	 x_j &\geq 0, \hat{a}_j \geq 0 ,&& \forall j \in [3] \nonumber \\
	 \hat{b}_{k,j} &\geq 0 ,&& \forall j,k \in [3]\nonumber \\
	 0 &\leq t^* \leq 1/2. \nonumber
	\end{align}
	\caption{The Quadratic Program with Quadratic Constraints to compute an upper bound for the $\ell_1$-loss for the case $t\leq 1/2$, $(-,-,+)$, $((b,r),(b,r))$. The inequalities referring to function $\hat{y}$ in Figure~\ref{fig:NLP:worst-case:zeroinoutcome} are now replaced by $6$ new inequalities, for the specified phantom pattern.}
	\label{fig:NLP:worst-case:special_examples}
	\end{figure}

\end{document}